\documentclass[conference]{IEEEtran}
\usepackage{cite}
\usepackage{amsmath,amssymb,amsfonts}
\usepackage{algorithmic}
\usepackage{graphicx}
\usepackage{bbding}
\usepackage{textcomp}
\usepackage{booktabs}
\usepackage{siunitx}  % 用于对齐数字

\usepackage[linesnumbered,ruled,vlined,boxed]{algorithm2e}
\usepackage{xcolor}
\def\BibTeX{{\rm B\kern-.05em{\sc i\kern-.025em b}\kern-.08em
    T\kern-.1667em\lower.7ex\hbox{E}\kern-.125emX}}
\usepackage{ragged2e}
\usepackage[inkscapelatex=false]{svg}
\usepackage{multirow}
\usepackage{booktabs}
\usepackage{caption}
\usepackage{subcaption}
\usepackage{diagbox}
\graphicspath{{Figures/}} 
\usepackage[hidelinks]{hyperref} % for refer jump -by yusong
\usepackage{amsthm}
\usepackage{multirow}
\usepackage{makecell}
\usepackage{float}
\usepackage{makecell, tabularx}
\usepackage{enumitem}
\usepackage{comment}
\usepackage{mdframed}
\usepackage{courier}
\usepackage{balance}
\usepackage{tcolorbox} % tcolorbox宏包用于创建带有颜色背景的框
\usepackage{hyperref}

\clearpage
\newcommand{\stitle}[1]{\vspace{1ex}\noindent{\bf #1}}
\newcommand{\etitle}[1]{\vspace{1ex}\noindent{\underline{\em #1}}}
\newcommand{\sstab}{\rule{0pt}{8pt}\\[-1.8ex]}

\newcommand{\zyf}[1]{\textcolor{blue}{#1}}
\newcommand{\ys}[1]{\textcolor{olive}{#1}}
\definecolor{zyjc}{rgb}{0,0.55,0.07}
\newcommand{\zyj}[1]{\textcolor{zyjc}{#1}}
\newcommand{\yaof}[1]{\textcolor{orange}{#1}}
\newcommand{\eat}[1]{}
\newcommand{\ie}{i.e.}

\newcommand{\red}[1]{\textcolor{red}{#1}}

\newcommand{\gray}[1]{\textcolor{gray}{#1}} % will be deleted
\newtheorem{lemma}{Lemma}

\newtheorem{theorem}{Theorem}

\newcommand{\GA}{\mathcal{G}}
\newcommand{\AP}{\mathcal{A}}
\newcommand{\SC}{\mathcal{S}}

\newcommand{\go}{GoGraph\xspace}
\theoremstyle{plain}

\SetKwInOut{KwInput}{Input}
\SetKwInOut{KwOutput}{Output}

\newcommand{\F}{\mathcal{F}}
\newcommand{\M}{\mathcal{M}}
\newcommand{\R}{\mathcal{R}}
\newcommand{\ord}{p}
\newcommand{\gr}[1]{\textcolor{black}{#1}}%gsf revise

\allowdisplaybreaks[4]

\IEEEoverridecommandlockouts % 脚注

\begin{document}

\title{Fast Iterative Graph Computing with Updated Neighbor States
{%\footnotesize \textsuperscript{*}Note: Sub-titles are not captured in Xplore and
%should not be used
}
}

\author{
Yijie Zhou$^{\dagger}$, 
Shufeng Gong$^{\dagger}$\thanks{The first two authors contributed equally to this paper.}, 
Feng Yao$^{\dagger}$, 
Hanzhang Chen$^{\dagger}$, 
Song Yu$^{\dagger}$, 
Pengxi Liu$^{\dagger}$, \\ 
Yanfeng Zhang\href{mailto:zhangyf@mail.neu.edu.cn}{\textsuperscript{\Envelope}}$^{\dagger}$\thanks{Yanfeng Zhang\href{mailto:zhangyf@mail.neu.edu.cn}{\textsuperscript{\Envelope}} is the corresponding author.}, 
Ge Yu$^{\dagger}$, 
Jeffrey Xu Yu$^{\ddagger}$ \\
\IEEEauthorblockA{%\textit{dept. name of organization (of Aff.)} \\
\textit{${\dagger}$Northeastern University, ${\ddagger}$The Chinese University of Hong Kong}\\
\{zhouyijie, yaofeng, chenhanzhang, yusong, liupengxi\}@stumail.neu.edu.cn, \\ \{gongsf, zhangyf, yuge\}@mail.neu.edu.cn, yu@se.cuhk.edu.hk}%
}

% \thanks{Manuscript received April 19, 2005; revised August 26, 2015.}

% \and
% \and

% \footnote{The first two authors contributed equally to this paper.}

% \footnote{Yanfeng Zhang is the corresponding author.}

% \input{response-2page}

\maketitle

% This document is a model and instructions for \LaTeX.
% This and the IEEEtran.cls file define the components of your paper [title, text, heads, etc.]. *CRITICAL: Do Not Use Symbols, Special Characters, Footnotes, 
% or Math in Paper Title or Abstract.
\begin{abstract}
Enhancing the efficiency of iterative computation on graphs has garnered considerable attention in both industry and academia.
Nonetheless, the majority of efforts focus on expediting iterative computation by minimizing the running time per iteration step, ignoring the optimization of the number of iteration rounds, which is a crucial aspect of iterative computation.
We experimentally verified the correlation between the vertex processing order and the number of iterative rounds, thus making it possible to reduce the number of execution rounds for iterative computation.
In this paper, we propose a graph reordering method, \go, which can construct a well-formed vertex processing order effectively reducing the number of iteration rounds and, consequently, accelerating iterative computation. Before delving into \go, a metric function is introduced to quantify the efficiency of vertex processing order in accelerating iterative computation.  
This metric reflects the quality of the processing order by counting the number of edges whose source precedes the destination.
\go employs a divide-and-conquer mindset to establish the vertex processing order by maximizing the value of the metric function. 
Our experimental results show that \go
outperforms current state-of-the-art reordering algorithms
by 1.83$\times$ on average (up to 3.34$\times$) in runtime. %, and 41\% on average (up to 71\%) reduction in iteration rounds.
% rabbit轮数平均是gograph的1.54倍，最高是2.33倍
\gr{Compared with traditional synchronous computation, our method improves the iterative computations up to 6.30$\times$ in runtime.}

\end{abstract}


\section{Introduction}\label{sec:intro}

\begin{comment}

    motivation:迭代计算好多都是固定轮数（sssp不是，而后边还用sssp举例子，想想怎么搞），但是其实可以减少轮数，当用最新邻居状态。我们发现这个对加速影响非常大，

    gauss-seidel迭代呢？迭代计算是算法常用形式，由于需要迭代，往往多轮才能收敛。数学中的多元一次不等式的解析一样，比如jacobi迭代。在多元一次公式中可以用gauss-seidel加速。采用最新的状态。事实上这种方法也可以应用在图计算中，而且已经有很多工作在探索。当可以异步执行时，往往就可以用。但是虽然他们尝试但是可能效果甚微，因为我们发现这个的效果跟顶点的处理顺序相关。
    
    总纲：故事需要改一下，不提动态调度了
    还是回到round-robin，
    是否需要引出：1）vertex-centric，为什么会有vertex-centric？图处理系统，采用vertex-centric，为什么要有系统？这里需要说吗？；2）消息传播模型，
    事实上有一些算法可以利用当前顶点状态来更新自身状态,从而加速收敛。
    那么如果让更多的邻居已经被处理，能否加快收敛？答案是肯定的，用一个例子。
    本文提出了一种顶点重排序方法，
    高斯赛德尔迭代扯进来，用图解释下（sssp），
    最好是intro的时候别提delta，证明的时候用delta，想想preliminary怎么写。

    正文总体分成两部分，一个是objective function，就是用什么来指导排序，衡量排序的标准，而且有个证明，另一部分是排序方法，还有分布式。
    
\end{comment}

% Graph analysis has been widely applied nowadays, like social recommendation
% systems~\cite{fan2019graph,wu2018socialgcn}, fraud detection systems~\cite{kurshan2020graph,sarma2020bank}, and COVID-19 spread detection and prevention~\cite{ahmed2020covid,bhapkar2020virus}. 
Iterative computation is an important method for 
%mining information from graphs 
graph mining, such as PageRank\cite{page1999pagerank}, single source shortest path (SSSP), breadth first search (BFS), and so on.
Since iterative computation involves traversing the entire graph multiple times to update vertex states until convergence, which is time-consuming. There have been many works that improve iterative computation from practice to theory \cite{low2012distributed,gonzalez2012powergraph,zhu2016gemini,galois,fan2017grape,yi2022flag,wu2022link}. \eat{For example, graph processing systems such as \red{GraphLab, PowerGraph, Gemini, Galois, and GRAPE} accelerate iterative computations by fully making full use of hardware performance and reducing the time of each iteration, while algorithms such as \red{PageRank and SSSP} theoretically accelerate iterative computations by obtaining an approximate result with a lower error rate.}

%These methods either accelerate the iterative computations by making full use of hardware or design case-by-case optimizations for a certain algorithm. 
Most of them %existing graph iterative computations
\cite{zhu2016gemini,fan2017grape,macko2015llama,liu2023mbfgraph} % ys:这里和后面对pagerank固定轮数的说法是不是不准确？只要是delta-base的或者基于活跃顶点的他们也不是根据最大轮判断收敛，例如Galois这类基于活跃顶点以及咋们自己和华科这类delta-base的论文感觉盖不住, 建议改成大部分工作
% \red{(there better be some citations here. I know Gemini is)}
% \ys{(Grape,llama, MBFGraph,...)} % ys: 感觉只要不是delta的应该都可以，明确说的也不少: llama, MBFGraph, Grape
accelerates iterative computations by reducing the runtime of each iteration round under the assumption that, for a given graph and iterative algorithm, the number of iteration rounds required to attain a specific convergence state is fixed.\eat{, because they believe that for a given graph and iterative algorithm, the number of iteration rounds is fixed.} For example, in PageRank\cite{brin1998anatomy}, it may take 20-30 iterations to converge; in SSSP, the number of iteration rounds is roughly %the length of \eat{remove the length of} 
the graph diameter \cite{sigact18ssspprove}.
% \eat{In fact, for a given graph and iterative algorithm, \eat{if the vertex update method and the processing order of the vertices are different, the number of iteration rounds will vary differently}% \zyj{the number of iteration rounds will vary depending on the vertex update method and the processing order of the vertices}.}
However, in practice, the number of iteration rounds for a given graph and iterative algorithm may exhibit substantial variation depending on the vertex update method.
We elaborate on this point in the following observation.
% , e.g., \textit{synchronous} or \textit{asynchronous}, and the 
% %order in which the system schedules vertices for processing
% \textit{vertex processing order}\cite{bertsekas2015parallel,xie2015sync,fan2020adaptive}. \zyf{need citations in many places [checked]}

\eat{However, hardware resources have been fully utilized by the current graph processing systems, and some work has even looked for new hardware resources to accelerate iterative calculations, such as \red{GPUs, FPGAs}, etc. While the optimizations of a specific algorithm are not able to be applied to other algorithms. \textit{Can we accelerate iterative computations by only performing simple changes on graphs themselves without reducing the utilization of hardware resources of graph processing systems and the accuracy of results?} It is very attractive for industry when the acceleration method is simple, effective, and cheap.} %\red{see comment} %说一下迭代轮数和runtime之间的影响关系

\stitle{Observation.} In general, traditional iterative algorithms are typically designed with a synchronous mode\cite{xie2015sync,fan2017grape}, processing each vertex in a round-robin fashion. %to update the states of vertices, in which % \gsf{introduce the messages, see comment}%gsf：顶点的更新根据邻居的状态【pregel】或者邻居的变化状态来更新自己【citeMaiter】，迭代计算中顶点的更新过程被抽象为顶点之间消息生成，收集与应用，既顶点将自身状态或者自身更新发送给邻居，邻居收集接受到的消息后更新自身状态。直到所有顶点的状态不再发生变化或者变化可忽略不计。\eat{When a vertex is computed during iterative processes, it typically obtains the state\cite{malewicz2010pregel} or the changes in the state\cite{zhang2013maiter} of its neighbors. This process can be abstracted as message generation, collection, and application in iterative computations. In other words, vertices send parts of their own states or updates to their neighbors, who use this information to update the corresponding items upon receiving the messages. This process continues until either all vertex states no longer change, or the change becomes negligible. However, asynchronous approaches can accelerate the process by not only eliminating synchronization barriers but also allowing the use of the latest vertex state and messages. In addition, the order of vertex processing can be arbitrarily scheduled.}
Specifically, to simplify the parallel semantics, in each round, every vertex is updated based on the state of its neighbors from the previous iteration. 
%The iterative algorithms converge when the state values of all the vertices are stable or change negligibly.
% In each round of iteration, the state of each vertex is updated based on the state values of neighbors, which are uniformly updated after the completion of the previous round 
% During synchronous iterations, vertex computation relies on the state values of its neighbors, which are updated in the previous iteration.
% \gsf{why the states are uniformly updated?}. %Convergence of the iterative algorithm occurs when all vertex states have either stabilized or exhibited negligible change.
% However, using the neighbor's state after the previous iteration to update the vertex state, the status of the vertices is updated slowly. 
%However, updating the vertex state only based on the \eat{neighbor's}\zyj{neighbors'} state from the previous round leads to slow convergence of vertex states. %\gsf{see comment} %gsf:在某些算法中，顶点可以采用异步更新方式，用其邻居当前轮迭代的状态来更新自身状态，如果其邻居已经被更新，而不是用上一轮迭代后邻居的状态。由于更新后的邻居更接近收敛状态，因此这种用邻居当前轮迭代的状态更新顶点状态的方式会加速迭代计算的收敛。【然后就可以上例子了】图1展示了采用不同迭代计算方式sssp算法的迭代计算过程和迭代轮数。1a是原始图，1b是采用同步迭代方法，每各顶点在更新时只能使用其邻居上一轮迭代后的状态。1c是采用异步迭代方法，每个顶点在更新时可以使用邻居已经更新后的状态。可以看出利用邻居最新状态可以更快达到收敛。
% In some algorithms, vertices can update asynchronously by utilizing the neighbors' states from the current iteration to update their own state, if their neighbors have already been updated, rather than using the state of neighbors from the previous iteration.
\eat{yaof: I think the async needs to be mentioned here because Fig 1 shows the async. gsf: I agree. Prof. Yu also mentioned this.} 
However, in some algorithms employing an asynchronous mode\cite{zhang2013maiter,fan2018adaptive}, vertices can update their state using the neighbors' state from the current iteration rather than the previous one.
This adjustment is beneficial because the updated neighbors' state is closer to convergence. Consequently, computations based on these updated state values yield results closer to convergence, 
thereby reducing the number of iteration rounds.
% thereby speeding up the convergence of the iterative computation.
%Moreover, the vertex processing order further affects the freshness of the neighbor states used for vertex updates.
Furthermore, we found that the vertex processing order further affects whether the latest state values can be utilized by its neighbors %in time for 
in the current iteration.
It means that the vertex processing order is significant for iterative computation to reduce the number of iteration rounds. %and accelerate the convergence speed.
% It means that the vertex processing order affects the convergence speed of iterative computations and the number of iteration rounds.

\begin{figure}
  \begin{minipage}{0.49\textwidth}
    \centering
    \includegraphics[width=0.88\textwidth]{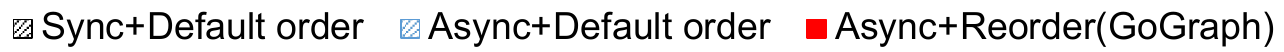}
    \\
    \subfloat[Runtime]{\includegraphics[width=0.5\textwidth]{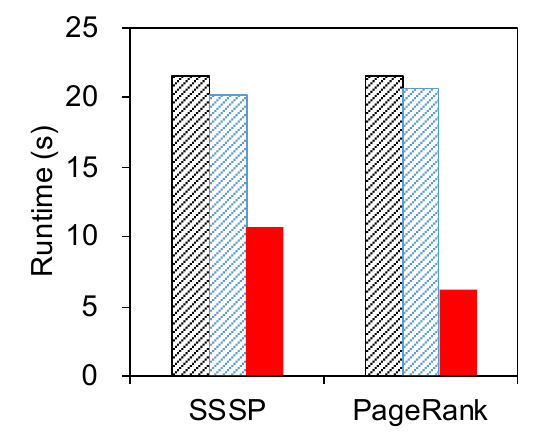}\label{fig:exam-time}}
    \subfloat[Iteration rounds]{\includegraphics[width=0.5\textwidth]{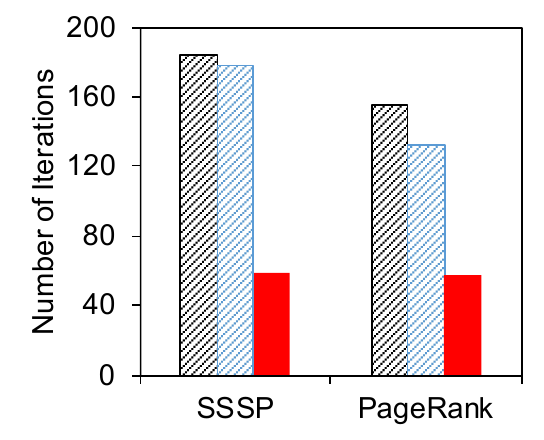}\label{fig:exam-round}}
    \vspace{-0.05in}
    \caption{Runtime \& Number of iterations of SSSP and PageRank with different vertex updating modes (Sync. vs. Async.) and different vertex processing orders (default order vs. reordered with \go) on wiki-2009 dataset \eat{suggest runtime as a figure (include SSSP and PageRank) and numter of iterations as a figure. because you want to compare different orders, but not runtime vs number. sync+rand order, async+rand order, async+reorder (GoGraph)}}\label{fig:exam-time-round}
  \end{minipage}
    % \vspace{-0.05in}
  \vspace{-0.2in}
\end{figure}

\begin{figure*}[htb]

\begin{minipage}{0.16\textwidth}
    \centering
    \includegraphics[width=\textwidth]{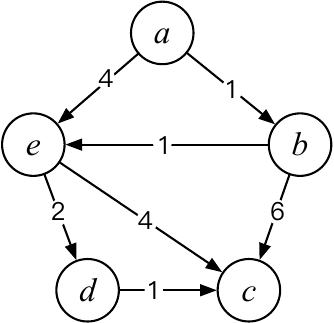}
    \subcaption{Origin graph}\label{fig:eg1-a}
\end{minipage}
\hfill
\begin{minipage}{0.28\textwidth}
    \centering
    \includegraphics[width=\textwidth]{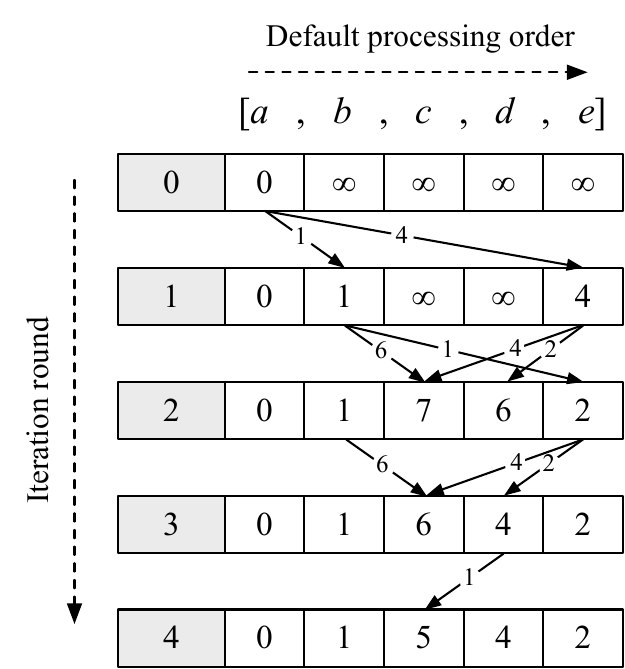}
    \subcaption{Sync. w/ default order \eat{Processing order, P should be lower case}}\label{fig:eg1-b}
\end{minipage}
\hfill
% \begin{minipage}{0.33\textwidth}
%     \centering
%     \includegraphics[width=\textwidth]{Figures/eg1-c.pdf}
%     \subcaption{Async iteration with runtime scheduling, Dijkstra's algorithm.}\label{fig:eg1-c}
% \end{minipage}
% \hfill
\begin{minipage}{0.22\textwidth}
    \centering
    \includegraphics[width=\textwidth]{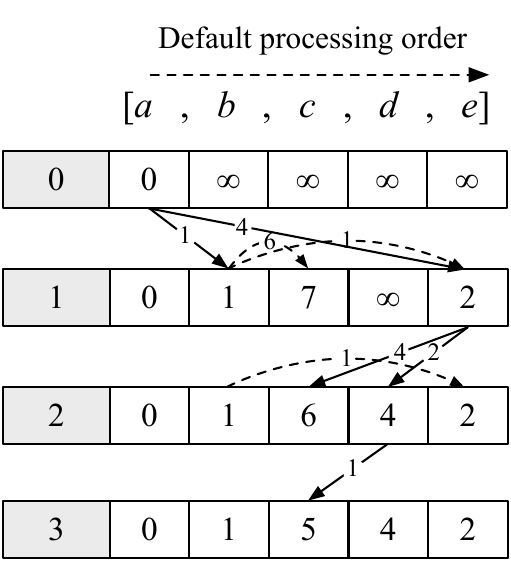}
    \subcaption{Async. w/ default order}\label{fig:eg1-c}
\end{minipage}
\hfill
\begin{minipage}{0.23\textwidth}
    \centering
    \includegraphics[width=\textwidth]{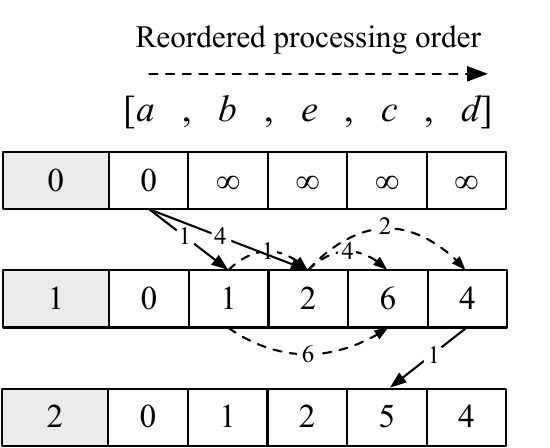}
    \subcaption{Async. w/ reordered order}\label{fig:eg1-d}
\end{minipage}
\vspace{-0.05in}
\caption{Iterative process and the number of iteration rounds generated by employing different iterative computation modes and vertex processing orders when running the SSSP algorithm, where the source vertex is $a$, the default processing order is alphabetical order based on vertex labels $[a, b, c, d, e]$, the reordered order is obtained by our method $[a, b, e, c, d]$\eat{what is original order? may need an arrow in the fig to indicate the order [checked. modified the figures.]}}\label{fig:eg1-all}
\vspace{-0.2in} %修改和文章的间距
\end{figure*}

Fig. \ref{fig:exam-time-round} shows the runtime and number of iterations of SSSP and PageRank with different vertex updating modes and different vertex processing orders on wiki-2009~\cite{nr}. Both 
%the runtime and the number of iterations of 
SSSP and PageRank have less runtime and fewer iterations in the asynchronous case. 
% \gsf{This is because, in asynchronous iterations, the updated state values of neighboring vertices in the current round can be utilized, which accelerates vertex convergence.} 
Furthermore, we reordered the vertices with \go (our vertex reordering method), ensuring that as many vertices as possible can leverage the latest state values of their neighbors from the current round. It can be observed that, after reordering vertices, the advantage of asynchronous mode is significantly improved, and the number of iteration rounds and runtime are significantly reduced.
We use Fig. \ref{fig:eg1-all} as an example to illustrate the results in Fig. \ref{fig:exam-time-round}, where the iterative algorithm is SSSP and the source vertex is $a$.
% \eat{should specify the source and algorithm [checked]} %It shows the iterative process and the number of iteration rounds generated by employing different iterative computation modes and vertex processing orders when running the SSSP algorithm, where the source vertex is $a$. 
Fig. \ref{fig:eg1-a} gives the topology of the graph.
% \eat{The initial processing order of vertices is $[a, b, c, d, e]$, and each vertex will be visited in sequence.
% We use the synchronous update mode in Fig. \ref{fig:eg1-d}, where each vertex is updated based on the values of its neighbors from the previous iteration. Convergence was reached in three rounds.}
In Fig. \ref{fig:eg1-b}, we employ synchronous iteration, where each vertex is updated using the values of its neighbors from the previous iteration. For example, in the first iteration, $e$'s state is updated to 4 based on the initial state of $a$, while in the second iteration, $e$ is updated to 2 based on the previous iteration state of $b$. It takes four iterations to achieve converged results.
Fig. \ref{fig:eg1-c} employs an iterative asynchronous technique in which each vertex uses the updated states of its neighbors if they have been updated. For example, in the first iteration, $e$'s state is updated to 2 based on the updated state of $b$ from the current iteration. Compared to the synchronous mode in Fig. \ref{fig:eg1-b}, $e$ changes and converges faster. Finally, the asynchronous mode achieves the same converged result in three rounds of iterations.
We observed that the updates for $c$ and $d$ depend on the state value of $e$. In Fig. \ref{fig:eg1-d}, we rearrange the processing order by placing $e$ before $c$ and $d$. With this adjustment, each vertex ensures that all of its incoming neighbors are updated before it, leading to the convergence of the graph after two iterations.

\stitle{Motivation}.
\eat{A natural question arises: since the convergence speed of the iterative computation can be improved by using the %previously %强调当前轮
updated neighbor states from the current iteration when updating the vertices, %would having more neighbors previously updated further accelerate the convergence of the vertex iteration algorithm? 
would it further accelerate the convergence % of the iterative computation 
\yaof{if more neighbors have been updated in the current iteration?}
% if there are more previously updated neighbors?
Fig. \ref{fig:eg1-e} gives an affirmative answer. }
% When processing a vertex, if its neighbors have already been processed, the status of the neighbors in this iteration can be used; otherwise, only the status of the neighbors after the previous iteration can be used, which is equivalent to devolving into synchronous iteration. \gsf{rewrite the last sentence carefully, it shows different opinions with my Chinese comments.}
% \zyj{As demonstrated in the example depicted in Fig. \ref{fig:eg1-d}, if a vertex's dependent neighbor has already undergone an update during the current iteration, the vertex can iterate based on this updated state, progressing further towards the convergence state. In other words, when} 
From Fig. \ref{fig:exam-time-round} and Fig. \ref{fig:eg1-all}, we can see that the processing order has a significant impact on the efficiency of asynchronous iterative computation. An effective vertex processing order can accelerate iterative computation. This motivates us to search for the optimal processing order to speed up iterative computation. 
\stitle{Our goal}. 
As observed above, %when iterating asynchronously, 
we can rearrange the vertices %to alter their 
processing order so that when updating each vertex, its neighbors have already been processed and updated. Since the updated neighbor state from the current iteration make the vertex closer to the convergence state, the iterative computation will be accelerated. %thereby accelerating iterative computation. 
\eat{Specifically, this can be done by assigning smaller IDs to the vertices processed preferentially %a smaller vertex ID, 
and larger IDs to the vertices processed later, 
%a larger ID. 
then sort the vertices according to their IDs. As a result, the vertex processing order is rearranged. However, which vertices should be processed first and assigned smaller IDs, and which vertices should be processed later and assigned larger IDs is not a simple matter.}%However, vertex reordering is an NP-hard problem \cite{wei2016gorder}.
In this paper, we aim to investigate the design of a vertex reordering method capable of accelerating iterative computations while keeping other factors constant,
% In this paper, we will explore how to design a vertex reordering method that can significantly accelerate %the convergence of 
% iterative computations %when the reordered vertices are still processed in a round-robin manner 
% without changing \eat{others}\zyj{other factors}, 
% such as the graph processing systems or the accuracy of results.
such as the execution mode and task scheduling strategy of the graph processing systems, as well as result accuracy.
% \gr{meanwhile reduce the cache miss}
%\gr{In addition, reduce cache misses by improving graph locality to achieve further improvements.}

%gsf：challenges：为了设计一种有效的顶点重排序方法，首先需要解决如下几个问题：
%1、寻找一种可以衡量顶点处理序列好坏的标准，从而指导顶点排序方法的设计。在设计一个有效的顶点重排序方法时，对于一个给定的图数据，顶点处理顺序的好坏的标准非常重要。虽然运行时间与迭代轮数可以衡量一个顶点处理顺序的还坏，但是在未执行迭代计算之前就能衡量好坏是非常困难的。
%在本文中，一个有效的排序方法所得到的的顶点处理顺序相比于随机处理顺序显著减少迭代计算轮数。 
%2、根据顶点处理顺序衡量标准，设计一种顶点排序方法。顶点排序是个NP-hard问题【cite{Gorder}】，对于有N顶点的图，有O(N！)种处理序列。对于一个给定的衡量标准，找到最优的处理序列是不可能的。此外，直觉上，从Figure1的例子可以看出衡量好坏的标准应该与边的指向有关，既从一个ID较小的点指向ID较大的点，则处理ID较大的顶点会利用ID较小的顶点已经更新后的状态，从而加速算法收敛。而图中顶点之间的链接关系错综复杂，调整一个顶点的位置可能会导致大量边的指向发生变化。\yaof{看到指向下意识想到是有向边的指向关系，这里想说的是执行顺序和边的指向一致可以加速收敛的intuition吧。}
\stitle{GoGraph.} 
To achieve the above goal, we propose a graph reordering method, \go,  which can construct an efficient vertex processing order to accelerate the iterative computation. It has the following unique considerations and design.

\etitle{Measure the quality of the processing order.} Before formulating an effective reordering method, it is necessary to design a metric for quantifying the efficiency of the vertex processing order in accelerating iterative computation. While some intuitive metrics, such as runtime and the number of iteration rounds, can be regarded as benchmarks, they require completing iterative computations for meaningful results. Evaluating the quality of the vertex processing order becomes quite challenging in the absence of full iterative computation. Regarding the above problem, based on the theoretical guidance, we introduce a metric function \eat{\zyf{use the word metirc}}that counts the number of edges whose sources are in front of destinations, which in turn reflects the quality of the processing order. %Ultimately, an \eat{optimized \zyf{remove optimized}}objective function is proposed on the basis of the metric function.

\etitle{A divide-and-conquer method for vertex reordering.} 
%To find an optimal vertex processing order that enhances iterative computational efficiency, it is essential to design an efficient reordering method. 
It is not trivial to find an optimal vertex processing order that enhances iterative computational efficiency.
%However, 
We demonstrate \eat{found \zyf{do not use past tense, demonstrate?}}that finding the optimal processing order is %not only 
an NP-hard problem. %but also an NP-approximate problem. 
Therefore, directly establishing the optimal order for a given metric is impractical.
% it is impractical to directly determine the optimal order for a given metric. \eat{based on the whole graph \zyf{why you emphasize the whole graph?}.} 
\go adopts the divide-and-conquer mindset. 
Initially, high-degree vertices are extracted from the graph to minimize their impact on the localization decisions of numerous low-degree vertices.
% First, \zyf{a sentense should have a subjective} extract high-degree vertices from the graph to minimize their influence on the localization decisions of a large number of low-degree vertices. 
Subsequently, the remaining graph is divided into smaller subgraphs to simplify the complexity of reordering, with a focus on intra- and inter-subgraph perspectives.
%在划分的时候使每个子图连接性好，从而减少cache miss。
% \eat{When partitioning the graph, the vertices in each subgraph are connected closely, which can help to reduce the cache miss, since most graph algorithms always access the vertex and its neighbors together.}
% \eat{The reordering process involves determining the current optimal position for the processing order of vertices within a subgraph and the partitioned subgraphs.}
The reordering process involves finding an optimal position for each vertex in the vertex processing order within the subgraph, and an optimal position for each subgraph in the subgraph processing order.
This is achieved by maximizing the value of the metric function.
% \yaof{Reordering is performed by finding the optimal position of the processing order for the vertices within a partition and the partitioned subgraphs that maximizes the value of the metric function.}
% Finally, \zyf{no subjective}  the high-degree vertices are taken into account to and yield the complete vertex processing order. 
Finally, the high-degree vertices are taken into account to determine the complete vertex processing order.
We theoretically prove the effectiveness and efficiency of \go. 
% We theoretically prove that the value of the measure function $\M(\cdot)$ of the processing order obtained via \go is not less than $|E|/2$. \zyf{why should I care about this?}

To summarize, we make the following contributions:
%提出了一种衡量order的标准metric function，
%证明了reorder问题是np-hard的
%提出了一种intuition的reorder方法。
%证明了该方法的下界和复杂度
%实验验证了该方法优越性。

%提出了一种衡量order的标准，并证明其合理性
%证明了reorder问题是np-hard的。
%提出了一种intiution的方法，并给出了其结果下界
%实验验证了该方法的优越性

\sstab
1) \textit{Metric Function}. A novel metric function $\M(\cdot)$ is proposed to measure the efficiency of processing order in accelerating iterative computation. We theoretically prove the effectiveness of $\M(\cdot)$ (Section \ref{sec:state}). 

%\sstab
%2) \textit{NP-hard \& NP-approximate}. We theoretically prove that finding the optimal processing order that maximizes the metric function value is an NP-hard and NP-approximate problem (Section \ref{sec:state}).

\sstab
2) \textit{\go}. An heuristic \eat{intuitive?heuristic?}graph reordering method \go is proposed, that can accelerate iterative computation significantly. We also provide an efficient implementation of \go %And \zyf{do not use And in the beginining of a sentence} 
%theoretically proved \zyf{which is subjective?} its effectiveness and efficiency
(Section \ref{sec:gograph}).

\sstab
3) \textit{Evaluation}. We evaluate the effectiveness of \go in accelerating iterative computation with a comprehensive experiment % to evaluate the practicality of \go 
(Section \ref{sec:expr}). \eat{a sentence should have a subjective} %It is shown that \go can achieve up to 3.33× speedup (2.09× on average) over the default processing order in asynchronous mode.
% , and \go only takes \red{56}\% time of other reorder methods.
 
\eat{should show the effectiveness and results before the organization description}

Furthermore, we first provide some preliminary foundational knowledge and some based definitions in Section \ref{sec:pre}. And, finally, we list a comprehensive related work in Section \ref{sec:relate} and conclude this paper in Section \ref{sec:conclu}.

\section{Preliminaries}\label{sec:pre}

This section will introduce the fundamental knowledge related to our work. 

\stitle{Graphs}. Consider a directed graph denoted as $G(V, E)$, where $V$ and $E$ are the set of vertices and edges. 
Given a vertex $v$, $IN(v)=\{u \mid (u,v) \in E\}$ and $OUT(v)=\{w \mid (v,w) \in E\}$ represent the set of incoming neighbors (abbr. \textit{in-neighbors}) and  outgoing neighbors (abbr. \textit{out-neighbors}) of $v$, respectively,
%$IN(v)=\{u \mid (u,v) \in E\}$ is the incoming neighbors set of vertex $v$, %$OUT(v)=\{w \mid (v,w) \in E\}$ is the outgoing neighbors set of vertex $v$, 
$|IN(v)|$ and $|OUT(v)|$ are the number of incoming neighbors (abbr. \textit{in-degree}) and outgoing neighbors (abbr. \textit{out-degree}), respectively. 
% where $V$ and $E$ are the sets of vertices and edges, respectively, and $P_E$ represents the properties of the edges in $G$. 
% We use $IN(v)=\{u|u, v\in V, (u,v) \in E\}$ and $d_{in}(v)$ to respectively denote the set of incoming neighbors and the in-degree of the vertex $v$. 
% When there is no ambiguity, we use $N_{in}(v)$ and $D_{in}(v)$ to denote these entities, respectively. 
%For a given vertex $v$, we use $\mathcal{V}_{in}(v)=\{u|u, v\in V, (u,v) \in E\}$ and $d_{in}(v)$ to denote the set of incoming neighbors (abbr. \textbf{in-neighbors}) and the in-degree of $v$.
%Similarly, $\mathcal{V}_{out}(v)=\{w|v, w\in V, (v,w) \in E\}$ and $d_{out}(v)$ denote the set of outgoing neighbors (abbr. \textbf{out-neighbors}) and the out-degree of $v$.

\stitle{Vertex processing order}. The vertex processing order $O_V=[v_0, \cdots, v_{|V|-1}]$ is one permutation of the vertices $V$ in $G$. There are $|V|!$ permutations for the graph with $|V|$ vertices. We aim to find one vertex permutation that accelerates the iterative computation as much as possible. 

\stitle{Ordinal number}. We define the \textit{ordinal number} of a vertex $v$ as its position in the processing order, denoted as $\ord(v)$. The value of $\ord(v)$ ranges from 0 to $|V|-1$.
\eat{, where $N=|V|$ is the number of vertices. }
In each iteration, if a vertex $u$ is processed before vertex $v$, then its ordinal number $\ord(u)$ is smaller than that of $v$, i.e., $\ord(u) < \ord(v)$, indicating that a lower ordinal number corresponds to an earlier processing order. For example, the ordinal numbers corresponding to the five vertices $a$, $b$, $c$, $d$, $e$ in Fig. \ref{fig:eg1-d} are 0, 1, 2, 3, 4, with $\ord(e)=4$. However, in the processing order shown in Fig. \ref{fig:eg1-d}, as the vertices are reordered, $\ord(e)$ is changed to 2.
% the rank of $b$ in the processing order in Fig. \ref{fig:eg1-d} is 1, while in the processing order in Fig. \ref{fig:eg1-e} is 4. \ys{There are many uses of r(u) later, should it be d(u)?}
% Each vertex has an initial state $S_i$, which continuously changes due to the flow of messages during iteration until the graph reaches the convergence state $S_c$.

\stitle{Positive/Negative edge}. An edge $(u, v)$ is defined as a \textit{positive edge} if the ordinal number of the source $u$ is smaller than that of the destination $v$, i.e., $\ord(u) < \ord(v)$. Otherwise, edge $(u, v)$ is a \textit{negative edge}, i.e., $\ord(u) > \ord(v)$. \eat{If the rank of the source is smaller than that of the destination,}In the case of the positive edges, the source has been updated when the destination is \eat{updated}being processed. Using the updated states of the source to update the destination may speed up the convergence of the destination. On the contrary, in the case of the negative edges, when the destination is updated, \eat{we have to use the state of the source from the previous iteration.}the source has not been updated yet, so it has to use the state of the source from the previous iteration.%如果一条边的source的rank小于destination的rank，我们就称为postive edge。否则就称为negative edge。因为如果source的rank小于destination的rank，则在更新destination的时候，source已经更新.利用source更新后的状态来更新destination，往往会加快destination的收敛速度。反之，destination更新的时候只能采用source在上一轮迭代的状态。

\stitle{Iterative Computation}.
%在迭代计算的没一轮中，每一轮对每个顶点处理一次。每个点根据其邻居的状态，利用更新函数F（）更新自己的状态，其中F（）的输入为顶点入邻居状态的集合。第k轮迭代计算是，顶点v的状态x_v的更新过程可表示为如下所示。
In each round of iterative computation, each vertex is updated using function $\F(\cdot)$, where the input of $\F(\cdot)$ is the set of the states of its incoming neighbors. In the $k$-th round of iteration, the update of the vertex $v$ can be expressed as follows,
\begin{comment}
 In iterative computation, the GAS (Gather-Apply-Scatter) model \cite{gonzalez2012powergraph} is commonly employed to depict the vertex %computation 
process \eat{computational process} during each iteration step. %The GAS model abstracts the semantics of the iterative graph computation 
It consists of three vertex-centric functions, \eat{includes three basic operations:} \textit{Gather} ($\GA$), \textit{Apply} ($\AP$), and \textit{Scatter} ($\SC$).
% These three operations constitute one iterative computation process, and each iteration updates the state of all vertices. 
The \textit{Gather} operation collects the messages sent to a vertex by its in-neighbors, the \textit{Apply} operation updates the state of the vertex using the gathered messages, and the \textit{Scatter} operation calculates the values to be sent to each of the out-neighbors based on the updated state value of the vertex. % Combining the GAS model and the delta-based model, the synchronous iterative process can be described as follows:
%Using the GAS model to summarize the iterative synchronous process is as follows:
In the GAS model, the iterative computation can be expressed as follows,
% \gsf{in the Graphs definition, you use $u$ or $v$ to represent vertices, why use $i$ and $j$ in this equation?}   
\end{comment}
\begin{equation}\label{eq:syn}
    \begin{aligned} 
        % M_v^k &= \GA(\{m_{u,v}^{k-1} | u \in IN(v)\}), \label{gas:g}\\
        % x_v^k &= \AP(x_v^{k-1}, M_v^k), \label{gas:a}\\
        % m_{v, w}^{k}&=\SC(x_v^{k}, \{w|w\in OUT(v)\}), \label{gas:s}
        x^k_v=\F\big(\{x^{k-1}_u| u\in IN(v)\}\big),
    \end{aligned}
\end{equation}
% where $i$ is the incoming neighbor of $j$, and $w$ is the out-neighbor of $j$ 
where %$u$, $v$, and $w$ are connected by directed edges, and $u$ is an in-neighbor of $v$, and $w$ is an out-neighbor of $v$. 
%$m_{u,v}^{k-1}$ is the message from $u$ to $v$ along edge $(u,v)$ during $(k-1)$-th round of iteration, %the connecting edges, based on the results of the $(k-1)$-th round of iterations. 
%$M_v^k$ is the collection of messages received by $v$ in the $k$-th iteration, 
$x_v^k$ is the state value of $v$ after the $k$-th iteration.
When processing vertex $v$, it accumulates the most recent state values of its in-neighbors $u$ and subsequently applies the function $\F(\cdot)$ based on these collected values. In synchronous iteration, the state values of $u$ are uniformly updated at the end of each round, so the state values of $u$ that vertex $v$ gathers in the $k$-th round originate from the $k-1$-th round. To illustrate this iterative process, we will employ the PageRank and SSSP algorithms as specific examples.
% The \ys{pseudocode} %下面还算为代码吗 
% listed below outlines the iterative computation process for the PageRank and SSSP algorithms in the $k$-th iteration.
% \zyj{Below, we use the PageRank and SSSP algorithms as specific examples to describe this iterative process.}
% \ys{If we want to use examples of pr and sssp to demonstrate the two writing methods, do we need to explain the algorithm a little? In addition, is the current way of writing async algorithms common? Is it delta-based? Need an explanation?} % ys：貌似页数够，可以详细解释下?

% \red{add example and pseudocode of iterative computation, pagerank and sssp}

\begin{figure}[ht]
\vspace{-0.1in}
\small
    \begin{minipage}[t]{0.22\linewidth}
        \texttt{PageRank:}
    \end{minipage}%
    \begin{minipage}[t]{0.75\linewidth}
        \raggedright
        $x_v^k = \sum\limits_{u\in IN(v)} x_u^{k-1} \cdot d / OUT(u)$, \\
        \texttt{s.t., $d$ is the damping factor.}
    \end{minipage}
\end{figure}
% \vspace{-\baselineskip}
\vspace{-0.2in}
\begin{figure}[ht]
\small
    \begin{minipage}[t]{0.22\linewidth}
        \texttt{SSSP:}
    \end{minipage}%
    \begin{minipage}[t]{0.75\linewidth}
        \raggedright
        \setlength{\baselineskip}{1.1\baselineskip}
        $x_v^k = \min \{x_v^{k-1}, x_u^{k-1}+d(u,v)| u\in IN(v) \}$, \\
        \texttt{s.t., $d(u,v)$ is the distance between $u$ and $v$.}
    \end{minipage}
    \vspace{-0.1in}
\end{figure}

As discussed in Section \ref{sec:intro}, using the updated state of the neighbors from the current iteration can accelerate iterative computation. %In other words, we can use messages sent by neighbors that have been processed in the current iteration to accelerate iterative computation. 
In the case of neighbors that have not yet been processed, we continue to use their states from the previous iteration. The update function in Eq. \ref{eq:syn} can be reformulated as Eq. \ref{eq:asyn}.
\eat{If each vertex uses its neighbors' state values that have been updated in the current iteration, rather than all the values from the previous round, then
% since each vertex updates its own state immediately upon receiving a message, 
the messages received in the \textit{Gather} step of Equation \ref{eq:gas} do not solely originate from $m_{u,v}^{k-1}$.
% the previous iteration round.
In the $k$-th iteration round, 
we use $u'$ to represent the in-neighbors that processed before $v$, and
% if vertex $u_1$ is processed prior to $v$, 
$u_1$'s state value will be updated in the $k$-th round prior to $v$. The message $v$ \textit{gathers} partly originate from $u_1$, denoted as $m_{u_1,v}^{k}$, and the other messages are from the in-neighbors that are processed after $v$, i.e., $u_2$, and the messages are denoted as $m_{u_2,v}^{k-1}$.}
% Conversely, if the processing order of $u_2$ comes after $v$, when $v$ \textit{gathers} a message from $u_2$, it utilizes the state value that $u_2$ updated in the previous iteration, denoted as $m_{u_2,v}^{k-1}$. 

\begin{equation}\label{eq:asyn}
\begin{aligned}
x^k_v=\F\big(&\{x^{k}_{u_1}| u_1\in IN(v), \ord(u_1) < \ord(v)\} \\
\cup &\{x^{k-1}_{u_2}| u_2\in IN(v), \ord(u_2) > \ord(v)\}\big).
% M^{k}_v = \GA \big(&\{m^{k}_{u,v} |r(u) < r(v), u \in IN(v)\} \\ 
% \cup &\{m^{k-1}_{u,v} |r(u) > r(v), u \in IN(v)\} \big). \\
%  x_j^k &= \AP(x_j^{k-1}, m^{k}_j), \\
% m^{k}_{j,w}&=\SC(m^{k-1}_j, P_{j,w})
\end{aligned}
\end{equation}
In the $k$-th round of the Eq. \ref{eq:asyn}, $u_1$ with ordinal numbers smaller than $v$ are updated, while $u_2$ with ordinal numbers larger than $v$ remain unchanged.
% where the $u$ whose ordinal number is smaller than $v$ is updated in the $k$-th round, while the $u$ whose ordinal number is larger than $v$ remains unchanged.
The vertex state update examples for PageRank and SSSP can also be modified accordingly:
\begin{figure}[ht]
\vspace{-0.1in}
\small
    \begin{minipage}[t]{0.22\linewidth}
        \texttt{PageRank:}
    \end{minipage}%
    \begin{minipage}[t]{0.75\linewidth}
        \raggedright
        \setlength{\baselineskip}{1.1\baselineskip}
        $x_v^k = \sum\limits_{u_1\in IN(v), p(u_1) < p(v)} x_{u_1}^{k} \cdot d / OUT({u_1}) + \sum\limits_{u_2\in IN(v), p(u_2) > p(v)} x_{u_2}^{k-1} \cdot d / OUT({u_2})$, \\
        \texttt{s.t., $d$ is the damping factor%; $u_1$ is processed before $v$ and $u_2$ % u1/u2改成i/j,u1/u2看着不像变量吧
        %is after $v$.
        }
    \end{minipage}
\end{figure}
 \vspace{-0.2in}
\begin{figure}[ht]
\small
    \begin{minipage}[t]{0.22\linewidth}
        \texttt{SSSP:}
    \end{minipage}%
    \begin{minipage}[t]{0.75\linewidth}
        \raggedright
        \setlength{\baselineskip}{1.1\baselineskip}
        $x_v^k = \min \{x_v^{k-1}, x_{u_1}^{k}+d(u,v), x_{u_2}^{k-1}+d(u,v) |$\\
        $u_1, u_2\in IN(v), p(u_1) < p(v), p(u_2) > p(v)\}$, \\
        \texttt{s.t., $d(u,v)$ is the distance between $u$ and $v$; $u_1$ is processed before $v$ and $u_2$ is after $v$.}
    \end{minipage}
    \vspace{-0.1in}
\end{figure}

In practice, the vertex update method in Eq. \ref{eq:asyn} has been applied in many fields, such as Gauss-Seidel iteration \cite{koester1994parallel} in linear algebra, \gr{Adsorption \cite{adsorption}, Katz metric \cite{katz1953new}, SimRank \cite{jeh2002simrank}, Belief propagation \cite{belief} and so on \cite{wang2020powerlog, yao2023ragraph}}.%Gauss-Seidel PageRank \cite{silvestre2018pagerank} in graph theory. 

\section{Problem Statement}\label{sec:state}

%gsf：本章对本文所研究的问题进行了形式化定义，并结合本文研究目标，提出了一种衡量顶点处理顺序质量好坏的标准,并从理论上解释了该标准的合理性。

%gsf:在设计重排序方法之前，首先对该问题进行了形式化定义，并提出了一个目标函数。该目标函数可以用来衡量顶点处理顺序的好坏。当目标函数值取最大时，其处理顺序为最优顺序。
%为什么需要objective function？因为有了objective function，我就可以去设计某个方法，使整个objective function的值得到最优。

\eat{It looks like the problem has not been formally defined. Define the problem with x* and their sub}

%在介绍问题之前，先引出F需要具备的性质，然后针对其性质对问题进行定义，然后提出了衡量标准，基于衡量标准提出本文的objective function
Before introducing the vertex reordering method, we first clarify the properties that function $\F(\cdot)$ required to have.
Next, we formally define the vertex reordering problem and introduce a metric function designed to assess the efficiency of processing order in accelerating iterative computation.
% Then we define the problem of reordering formally, we propose a measure function to measure the efficiency of processing order in accelerating iterative computations. 
Finally, we propose the objective function of our paper based on the metric function. 
%The proposed objective function can be used to measure the quality of the vertex processing order. The optimal processing order will result in the largest objective function value.
%When the objective function value is the largest, the processing order is the optimal order.

%在迭代计算中，像pagerank，sssp等这种fixed point且可异步执行的算法，在迭代计算时利用已经更新后的状态才能加速收敛。而对于非fixed point的算法，比如the training of neural network等迭代计算，顶点状态会发生抖动，此时利用最新的状态并不代表接近收敛状态，因此未必会加速。事实上，使用邻居最新状态能够加速需要迭代算法需要满足monotonic的特性，既F是monotonic function。当满足monotonic的时候，顶点的状态会一直朝着收敛状态走，邻居接近收敛状态，则顶点也会更接近收敛状态。既顶点更新函数F12需要满足如下不等式（单调性公式）。
%Iterative algorithms such as PageRank and SSSP, which are fixed-point and can be executed asynchronously, will be accelerated if the incoming neighbors have been updated before each vertex is processed within the same round. %In the case of iterative computations such as the training of neural networks, the states of vertices tend to fluctuate. The updated states of neighbors may become further away from the converged states. Therefore, in this case, it is uncertain whether using the latest state of neighbors can speed up the convergence of iterative computations. 

In fact, if utilizing the updated state of neighbors to accelerate the iterative algorithm, the vertex states change monotonically increasing or decreasing, i.e., the update function $\F(\cdot)$ should be a monotonically increasing function. %给定一个单调递增函数，顶点的状态会逐渐减小（逐渐增大），并趋近于收敛状态。在顶点更新过程中，顶点的状态会变的更大或更小，并更趋近于收敛状态，当邻居的状态更小或（更大）时。

\stitle{Monotonic}. In iterative computation, considering a monotonically increasing function $\F(\cdot)$, the vertex states progressively decrease (or increase), moving closer to a convergent state. During the update of each vertex, its state value diminishes (or grows), advancing towards convergence, especially when the state value of its neighbors is smaller (or larger).
% In iterative computation, given a monotonically increasing function $\F(\cdot)$, the state of the vertices gradually decreases (increases) and approaches a convergent state. When updating each vertex, the state value of the vertex will become smaller (larger), and progress toward the convergence state, if the state value of the neighbor is smaller (larger). \eat{The function is a monotonic function if a larger input results in no smaller value, or a smaller input results in no larger value.} 
Then, we have the following inequality 
\begin{equation}\label{eq:mot}
\begin{aligned}
    &\begin{aligned}
        &\check{x}_v=\F(x_{u_1}, \cdots \check{x}_{u_i}, \cdots x_{u_{|IN(v)|}}) \\
        &\leq \hat{x}_v=\F(x_{u_1}, \cdots \hat{x}_{u_i}, \cdots x_{u_{|IN(v)|}}) 
    \end{aligned} 
 \\ &{if} \quad \check{x}_{u_i} \leq \hat{x}_{u_i}
 \end{aligned}
\end{equation}
where $\{u_1, \cdots u_i, \cdots u_{IN(v)}\}$ is the set of $v$'s incoming neighbors, $\check{x}_v, \check{x}_{u_i}$ and $\hat{x}_v, \hat{x}_{u_i}$ are two state values of $v$ and $u_i$ respectively. 

There is a broad range of iterative algorithms with monotonic vertex update functions, including SSSP \eat{(Single-Source Shortest Path)}, Connected Components (CC)\cite{hsu1994cc}, Single-Source Weighted Shortest Path (SSWP)\cite{yang2020sswp}, PageRank, Penalized Hitting Probability (PHP)\cite{wu2014php}, Adsorption\cite{baluja2008Adsorption}, and more. Monotonic property provides a wider optimization space for iterative computation and has been emphasized in many graph analysis works\cite{feng2021risgraph,vora2017kickstarter,wang2020powerlog,zhang2013maiter}. % \ys{Add some citations: risgraph,kickstater,PowerLog,maiter} %这里是否也需要说下有很多论文和系统在研究这些算法，例如risgrph,kickstater,PowerLog,maiter等. 
We also organize our vertex reordering method according to the monotonic function $\F(\cdot)$. The subsequent lemma is derived directly from the monotonicity of $\F(\cdot)$.
% Our vertex reordering method is also organized based on the monotonic function $\F(\cdot)$.

% Based on the monotonic function $\F(\cdot)$, we have the following Lemma.
\begin{lemma}\label{lem:mont}
    Given the monotonically increasing function $\F(\cdot)$ of an iterative algorithm, $\check{x}_v=\F(x_{u_1}, \cdots \check{x}_{u_i}, \cdots x_{u_{|IN(v)|}})$ and $\hat{x}_v=\F(x_{u_1}, \cdots \hat{x}_{u_i}, \cdots x_{u_{|IN(v)|}})$ are two state values of $v$, where $\{u_1, \cdots u_i, \cdots u_{|IN(v)|}\}$ are the incoming neighbor of $v$, $\check{x}_{u_i}$ and $\hat{x}_{u_i}$ are two state values of $i$-th incoming neighbors of $v$, $x^*_v$ and $x^*_{u_i}$ are the converged states of $v$ and $u_i$ respectively. Then we have \eat{$u_i$ is only one neighbor?[modified]}
    \begin{equation}\label{eq:deltmot}
        |x^*_v-\check{x}_v| \leq |x^*_v-\hat{x}_v| \quad if \quad |x^*_{u_i}-\check{x}_{u_i}| \leq |x^*_{u_i}-\hat{x}_{u_i}|
    \end{equation}
\end{lemma}

\begin{proof}
    There are two cases in the proof, 1) the vertex state continues to increase, and 2) the vertex state continues to decrease.
    
    When the state values of the vertices continue increasing during the iterative computation, we have $x^*_{u_i} \geq \textit{max}\{\check{x}_{u_i}, \hat{x}_{u_i}\}$ and $x^*_v \geq \textit{max}\{\check{x}_v, \hat{x}_v\}$. Since $|x^*_{u_i}-\check{x}_{u_i}| \leq |x^*_{u_i}-\hat{x}_{u_i}|$, then we have $\check{x}_{u_i} \geq \hat{x}_{u_i}$. According to inequation \ref{eq:mot}, we have $\check{x}_v \geq \hat{x}_v$. It means that $\check{x}_v$ is closer to the converged state, i.e., $|x^*_v-\check{x}_v| \leq |x^*_v-\hat{x}_v|$.

    Similarly, when the state values of the vertices continue decreasing during the iterative computation, we have $x^*_{u_i} \leq \textit{min}\{\check{x}_{u_i}, \hat{x}_{u_i}\}$ and $x^*_v \leq \textit{min}\{\check{x}_v, \hat{x}_v\}$. Since $|x^*_{u_i}-\check{x}_{u_i}| \leq |x^*_{u_i}-\hat{x}_{u_i}|$, then we have $\check{x}_{u_i} \leq \hat{x}_{u_i}$. According to inequation \ref{eq:mot}, we have $\check{x}_v \leq \hat{x}_v$. It means that $\check{x}_v$ is closer to the converged state, i.e., $|x^*_v-\check{x}_v| \leq |x^*_v-\hat{x}_v|$.
\end{proof}
%\stitle{Property} 

Based on the above lemma, we propose the following theorem to explain why using the updated state of incoming neighbors from the current iteration accelerates iterative computation.
%value of objective function $\M$ can be used as the metric to measure the quality of processing order.

%theorem：给定一个图G，和两个processing order$O_V1=[]$和$O_V2$，其中Ov1和Ov2的区别是顶点u和v的位置不同，且图中存在一条u指向v的一条边，根据\M定义公式可知\M(OV1)<\M(OV2).那么我们有如下结论|x*v-x‘v|<=|x*v-x''v|，其中x*是收敛状态，x‘是ov1在k轮之后的结果，而x''是ov2在k轮之后的结果
%proof:

\begin{theorem}\label{thm:mont}
    Given the graph $G$, monotonically increasing update function $\F(\cdot)$ and two processing orders $O^1_{V}=[a, \cdots v,u, \cdots z]$,  $O^2_{V}=[a, \cdots u,v, \cdots z]$. There is an edge $(u,v)$ and no edge $(v,u)$ in $G$. The difference between $O^1_{V}$ and $O^2_{V}$ is the ordinal number of $u$ and $v$, i.e., $\ord(u)>\ord(v)$ in $O^1_{V}$ and $\ord(u)<\ord(v)$ in $O^2_{V}$, which results in $O^2_{V}$ having one more positive edge than $O^1_{V}$. Then, we have 
    \begin{equation}
    \begin{aligned}
    \hat{x}^k_v&=\F_{O^1_V}(x^{k-1}_u, \cdots) \\
    \check{x}^k_v&=\F_{O^2_V}(x^{k}_u, \cdots)
    \end{aligned}
    \quad \text{and} \quad
    \begin{aligned}
        |x^*_v-\hat{x}^k_v| \geq |x^*_v-\check{x}^k_v|
    \end{aligned}
    \end{equation}
    where $x^*_v$ is the converged state of $v$, $\check{x}^k_v$ and $\hat{x}^k_v$ are the state values of $v$ after $k$-th iteration resulted from $O^1_{V}$ and $O^2_{V}$ respectively.
\end{theorem}

\begin{proof}
There is an edge $(u,v)$ and no edge $(v,u)$ in $G$. Therefore, compared with $O^1_{V}$, %there is an additional edge 
in $O^2_{V}$, the ordinal number of $v$ is larger than $u$ and $u$ has been updated before updating $v$. %than whose source's rank is larger than the destination's rank. %According to the definition of $\M$, we have $\M(O^1_{V})=\M(O^2_{V})-1$, i.e., $\M(O^1_{V})<\M(O^2_{V})$.
Thus when updating $v$, we can use the updated state of $u$ from the current iteration, i.e., $\hat{x}^k_v=\F_{O^1_V}(x^{k}_u, \cdots)$. While in $O^1_{V}$, we have to use the state of $u$ from the previous iteration, i.e., $\check{x}^k_v=\F_{O^2_V}(x^{k-1}_u, \cdots)$. 

Since $\F(\cdot)$ is a monotonically increasing function, the states of the vertices gradually progress toward the convergence states. The state of updated $u$ is closer to converged state, i.e., $|x^*_u-x^{k-1}_v|\geq |x^*_u-x^{k}_v|$. According to the Lemma \ref{lem:mont}, we have $|x^*_v-\hat{x}^k_v| \geq |x^*_v-\check{x}^k_v|$.
\end{proof}

The theorem presented above highlights that leveraging the updated state of incoming neighbors from the current iteration accelerates iterative computation. The more updated incoming neighbors are involved, corresponding to a higher count of positive edges, results in a faster convergence of the vertex. This implies that the speed of vertex convergence is intricately linked to the order in which vertices are processed. Consequently, the promotion of iterative computation can be achieved through the adoption of an efficient processing order, ultimately reducing the required number of iteration rounds. In summary, the problem addressed in this paper can be formulated as follows.
% From the above theorem, we can see that using the updated state of incoming neighbors from the current iteration accelerates iterative computations. The more updated incoming neighbors are used, i.e., the more positive edges, the faster the convergence of the vertex is. It means that the vertex convergence speed is affected by vertex processing order. Therefore, we can speed up iterative computations by employing an efficient processing order, thus reducing the number of iteration rounds. Finally, the problem of this paper can be formulated as follows.

%正如intro所述，本文加速迭代计算通过减少迭代计算轮数来实现，这可以通过改变顶点的处理顺序来实现，因此本文的目标是寻找一个有效的处理顺序来最小化迭代计算轮数k。即
\etitle{Problem Formulation.} Find an optimal graph processing order, denoted as $O_V=\R(G)$, that can minimize the number of iteration rounds $k$ required by the iterative algorithm when performing $\F(\cdot)$ with $O_V$, i.e.,
%As mentioned in Section \ref{sec:intro}, the primary objective of this paper is to accelerate iterative computations by minimizing the number of iterative computation rounds, which can be achieved by changing the vertex processing order. Therefore, the objective of this paper is to find an effective vertices processing order, denoted as $O_V$, that minimizes the number of iterative computation rounds $k$, 
\begin{equation}
    O_V=\underset{k}{\arg\min} (|X^*_V \ominus \big(X^k_V=\F^k_{\R(G)}(X^0_V)\big)| \leq \epsilon).
\end{equation}
where $X^*_V=\{x^*_v| v\in V\}$ is the state value set of converged vertex, $X^k_V=\F^k_{\R(G)}(X^0_V)$ is the set of vertex states after $k$ iterations using function $\F(\cdot)$ with $O_V=\R(G)$ as the processing order, $\R(G)$ is the graph reordering function that can return a permutation of vertices, $\epsilon \geq 0$ is the maximum tolerable difference between the vertex convergence state and the real state,
$\ominus$ computes the difference between $X^*_V$ and $X^k_V$.
In general, there are two implementations of $|X^*_V \ominus X^k_V|$, \romannumeral1) $\text{max}(|x^*_v - x^k_v|, v\in V)$, e.g., in SSSP, \romannumeral2) $\sum_{v\in V}(|x^*_v - x^k_v|)$, e.g., in PageRank.
%the overall state of vertices in $X_V$ after undergoing $k$ iterations and the overall state at convergence.

However, before iterative computation, it is impossible to know the number of iterations for a given processing order $O_V$. Thus, we are not able to know which processing order $O_V$ returned by $\R(G)$ is optimal. Therefore, it is required to design another metric to measure the quality of processing orders.

%上文所述，前邻居越多越好，由此我们可以提出如下的objective function。

As observed in Section \ref{sec:intro} and supported by Theorem \ref{thm:mont}, %when each vertex is processed, the more updated incoming neighbors it has, the fewer iterations and less runtime are required by the iterative computations. Based on this observation and intuition, we propose the following measure function.
the processing of each vertex shows that more updated incoming neighbors corresponds to fewer iteration rounds and shorter runtime in iterative computation. Based on this insight and intuition, we propose the following measure function.

\etitle{Metric Function}. Given a graph processing order $O_V$, we use the following measure function to measure the efficiency of processing order $O_V$ in accelerating iterative computation.
\begin{equation}\label{eq:objf}
\begin{aligned}
     &\begin{aligned}
    \M(O_V)=\sum_{v\in V}\sum_{u\in IN(v)}\chi(u,v)=\sum_{(u,v)\in E}\chi(u,v)  \\
    \end{aligned}
\\ &{where}
    \\ &\chi(u,v)=\left\{
    \begin{aligned}
    1 & \qquad if \quad p(u)<p(v), \\
    0 & \qquad if \quad p(u)>p(v). \\
    \end{aligned}
    \right.
\end{aligned}
\end{equation}

It can be seen that $\M(O_V)$ counts the sum of incoming neighbors with smaller ordinal numbers for each vertex. In other words, the value of $\M(O_V)$ is the number of edges where the source vertex has a smaller ordinal number than the destination vertex, indicating positive edges. The function $\M(O_V)$ achieves its maximum value when all vertices' incoming neighbors are in front of them, i.e., the incoming neighbors have smaller ordinal numbers. In this scenario, as all incoming neighbors are positioned in front of each vertex, updating each vertex ensures that its incoming neighbors have been updated and are closer to the convergence state. Leveraging the updated state of incoming neighbors propels the vertex closer to convergence, thereby accelerating the iterative computation.
% Because all incoming neighbors are in the front of each vertex, when updating each vertex, its incoming neighbors have been updated and are closer to the convergence state. Using the updated state of the incoming neighbors brings the vertex closer to the convergence state. Thereby the iterative computations are accelerated.

%Next, we will discuss whether $\M()$ can be used as the metric to measure the quality of vertex processing orders. Before that, we first introduce the property of function $F$.

\eat{I think: What we rely on is the monotonicity of $F$, which can be either increased(pagerank) or decreased(sssp). Although only single increase is proven before, the proof of single decrease is the same, but it should be mentioned here? And is it better to use monotonic $F$ instead of monotonically increasing in the following summary?}
%Therefore, when the update function $F$ is monotonically increasing, the vertex's state consistently progresses toward the convergence state. The vertex will be closer to the convergence state after updating with the neighbors' states if the neighbors are updated and closer to the convergence states. For a given vertex $v$, the more incoming neighbors of $v$ in the front of $v$ in the processing order, the greater value $\M()$ has and the faster $v$ converges. Therefore, we can use the value of $\M()$ to measure the effectiveness of processing order in accelerating the iterative computations.%对于给定点v，其入邻居在前面的越多，会导致\M的值越大，顶点收敛越快。因此，可以用\M值来衡量processing order对迭代计算加速的有效性。 
%\ys{Add an objective function?} % 在公示4或者这儿给个目标函数？gsf：对目标函数应该放在这。

Based on the measure function, we propose the following objective function.

\etitle{Objective Function}. Since the value of $\M(\cdot)$ serves as a metric to measure the efficiency of processing order in accelerating iterative computation, the objective of this paper is to identify the optimal graph processing reorder, denoted as $O_V=\R(G)$, that maximizes the $\M(\cdot)$ value. 
\begin{equation}
    % O_V=\arg\max\ \M(\R(G)),
    O_V=\underset{\R(G)}{\arg\max}\ \M(\R(G)),
\end{equation}
where $\R(G)$ returns a permutation of vertices in $V$.

% From the objective function, the aim of our work becomes to find the optimal processing order that maximizes the number of positive edges. From the following corollary, find the optimal processing order that maximizes the value of $\M(\cdot)$ is not only an NP-hard problem but also an NP-approximate problem.
\etitle{NP-hard \& NP-approximate}. Derived from the objective function, our goal can be seen as identifying an optimal processing order that maximizes the number of positive edges. %However, as indicated by the subsequent corollary, determining the processing order that maximizes the value of $\M(\cdot)$  is difficult work. %not only an NP-hard problem but also an NP-approximate problem. \eata{The sentence not only but also seems to be repeated many times, should we change the expression?}
\begin{comment}
    Given a directed graph $G$, the graph reordering and the objective function $\M(\R(G))$, finding the optimal vertex processing order $\R(G)$ that maximizes $\M(\R(G))$ value is not only an NP-hard problem but also an NP-approximate problem.
\end{comment}
%The value of $\M(\R(G))$ is the number of positive edges in the processing order $\R(G)$. 
This is similar to the \textit{topological sort}. The topological ordering of a directed graph is a linear ordering of its vertices such that for every directed edge $(u,v)$ from vertex $u$ to vertex $v$, $u$ comes before $v$ in the ordering. 
Therefore, if the graph is a directed acyclic graph, then we can use the topological sorting algorithm\cite{topological-sort} to reorder the vertices. 
% Then the $\M(\R(G))=|V|$ will be maximized because the topological sorting algorithm aims to find a linear ordering of vertices such that 
In this case, $\M(\R(G))$ will achieve the maximum value $|E|$. 
    
% However, there are a large number of cycles in real graph data \cite{boutilier2004cp,dechter1990enhancement,gottlob2000comparison}, thus it is impractical to use topological sorting to reorder the vertices directly. A workable way is to first generate a directed acyclic subgraph from the directed graph with cycles, which can be done by deleting some edges from the cyclic graph and then performing the topological sorting algorithm. Since the purpose of reordering is to get more positive edges, as many edges as possible should be preserved in the generated directed acyclic subgraph, which is called the Maximum Acyclic Subgraph (MAS) problem. It has been proved that the MAS is not only an NP-hard problem but also an NP-approximate problem \cite{karp2010reducibility,guruswami2008beating,lucan2016exploring}. Therefore, maximizing the value of $\M(\R(G))$ is also an NP-hard problem and an NP-approximate problem. 
%However, real-world graph data often contains a substantial number of cycles\cite{boutilier2004cp,dechter1990enhancement,gottlob2000comparison}, making it impractical to directly employ topological sorting for vertex reordering. 
A more viable approach involves first generating a directed acyclic subgraph from the cyclic graph by selectively removing edges, followed by the application of the topological sorting algorithm. Since the primary objective of reordering is to enhance the number of positive edges, preserving as many edges as possible in the generated directed acyclic subgraph is crucial. This problem is known as the Maximum Acyclic Subgraph (MAS) problem, and it has been demonstrated to be both NP-hard and NP-approximate\cite{karp2010reducibility,guruswami2008beating,lucan2016exploring}. Consequently, maximizing the value of $\M(\R(G))$ is inherently an NP-hard and NP-approximate problem.

In theory, \cite{guruswami2008beating} has proven that the lower bound of the number of edges that can be preserved in acyclic subgraphs is $|E|/2$. However, in practice, there have been many works \cite{hassin1994approximations,berger1997tight,charikar2007advantage,guruswami2008beating,cvetkovic2020maximal} that \eat{can find larger directed acyclic subgraphs.}have demonstrated the ability to identify larger directed acyclic subgraphs. Specifically, \cite{goldsmith2008computational,cormen2001introduction,liu2018cutting} derive a directed acyclic subgraph by deleting vertices or edges from Conditional Preference Networks, and \cite{cvetkovic2020maximal} \eat{finds}identifies the maximum acyclic subgraph based on matrix. \eat{However}Despite these achievements, these methods are not suitable for our problem. They are either designed for \eat{a specific scene}specific scenarios or \eat{cannot be used in}are impractical for large-scale graph data. 

% On the other hand, even if the maximum acyclic subgraph is obtained, and using topological sorting to obtain the processing order that maximizes $\M(\R(G))$. However, topological sorting algorithms often ignore the neighbor relationship between vertices, resulting in two connected vertices being far away from each other in processing order. When performing iterative computations, it is frequent to access the neighbors of each vertex. The CPU cache hit ratio will decrease when the vertex is far away from its neighbors \cite{wei2016gorder}. Therefore, we propose an efficient graph reordering method \go in this paper.

On the other hand, even if the maximum acyclic subgraph is obtained and topological sorting is employed to derive the processing order that maximizes $\M(\R(G))$, conventional topological sorting algorithms may overlook the neighbor relationships between vertices. This oversight can result in two connected vertices being positioned far apart from each other in the processing order. During iterative computation, frequent access to the neighbors of each vertex is common. The CPU cache hit ratio tends to decrease when a vertex is situated far away from its neighbors \cite{wei2016gorder}, which reduces computational efficiency. To overcome the above problems, %Therefore, 
in this paper, we introduce an efficient graph reordering method called \go.

\section{Reordering Method}\label{sec:gograph}

%According to the analysis in Section \ref{sec:state}, the value of $\M$ can be used as a metric to measure the efficiency of processing order in accelerating iterative computation. Therefore, \eat{the}\zyj{another} goal of this paper \eat{becomes}\zyj{is} to design a vertex reordering method to improve the $\M$ value of processing order. 
In this section, we introduce the vertex reordering algorithm, \go, an efficient and effective vertex reordering method. %Theoretically, the time complexity of \go is $O(E\cdot logE)$, and the lower bound of $\M(\cdot)$ value of the processing order produced by \go is $|E|/2$. \zyf{the sentence seems wrong, no predict}
%and theoretically ensure that the $\M(\cdot)$ value of the processing order obtained by \go is higher than that of a random processing order, i.e., the lower bound of the number of positive edges is $|E|/2$. Furthermore, \go is not only efficient but also effective, its time complexity is .

\subsection{\go}
Due to the complex links between vertices in the graph, the change in the ordinal number of a vertex may result in some changes in edges, \ie,  some positive edges become negative and some negative edges become positive, 
%position adjustment of a vertex in the processing order may cause changes in the ordinal number of the source and destination of a large number of edges, 
which result\eat{result?} in an unpredictable $\M(\cdot)$\eat{value is not necessary}. Therefore, it is difficult to reorder the vertex 
%processing order 
from the perspective of the whole graph. In \go, we adopt a divide-and-conquer method for vertex reordering.\eat{, which first divides the graph into subgraphs, then reorders vertices within and between subgraphs, and finally obtains the whole vertex processing order.} The method initially extracts the high-degree vertices from the graph, then divides the graph into subgraphs and reorders the vertices within and between these subgraphs, yielding the complete vertex processing order.

% This method initially divides the graph into subgraphs, then reorders vertices within and between these subgraphs, ultimately yielding the complete vertex processing order.

%算法过程：
\begin{comment}
    1、抽取大度点。度大的点？出度大的点？入度大的点？为什么要抽取度大的点？【度大的点往往存在很多边，包括入边和出边，先确定其位置还是后确定其位置？为什么？其位置确定比较困难，因此我们采用后确定位置的方式。另一个，由于边多，不管怎么排放，都可能会使一些边的source的rank小于destination的】
    2、做cluster，把数据局部化，全局考虑比较困难，因此局部排序比较好控制。why？因为局部情况下，可以很容易得到每个点的出邻居，入邻居等，而且量不大，大度点已经被排除。
    3、cluster看成vertex，cluster之间有权重的边，然后cluster规模小，也比较好排序。
    4、cluster之间和cluster内部排好序之后，还有两类点，一个是大度点，一个是因为大度点删掉后变成的孤立点。对于这两种先处理大度点，挑着度小的先来，插入已经排好序的order中，每个位置都可以算出一个分，插入使\M得分最高的地方【这里一个优化是大度点插入到cluster之间，然后再进入cluster内部插入】。然后在插入孤立点，也是插入使分最高的地方【这里的一个优化是插入完大度点之后，得到大度点的一个order，然后提取大度点order单独插入，这样就小了】。
    5、内部排序怎么排呢？选择score【出度减入度】最大的点作为初始点（第一个值一定是正的），然后执行bfs，在执行bfs时，优先访问score大的邻居。但是这样不能保证我得到的结果肯定比随机的好。所以，在每次bfs遍历往processing order中追加时，判断是否使\M的值增加现有边的一半，如果能，则加上，如果不能则调整插入位置，使其为一半以上。调整方法为从后往前查找找到第一个一半以上的位置。但是这个方法计算复杂，因此可以用score+bfs的方法，或者简单用score从小到大的方法。
    总结：以上方法本质上是一个批量边插入的方法。
\end{comment}

%算法的证明，证明该方法得到的\M值>=|E|/2，也就是比随机的好。
\begin{comment}
    lemma：给定一个processing order，来一个点插入进去，能保证跟这个点所关联的新插入的边，一半以上朝后。
    proof：极端情况下，放在第一个和最后一个，第一个的值为出度-入度，最后的值为入度-出度，要么等于0，要么一正一负，放在正的地方就可以保证一半以上朝后。
\end{comment}
\begin{comment}
    theorem：该方法得到的\M值>=|E|/2
    proof：第5步能够保证在每个cluster内部\M值大于等于一半的边。因为每插入一个点，根据lemma保证所关联的现有边能有大于等于一半的边朝后。
    同时cluster之间也是这么排序，因此能够保证cluster之间的边有一半及其以上朝后。
    大度点和孤立点插入一样的插入方法，因此能够保证一半及其以上的边朝后。
    所有情况加起来，就是至少一半以上的边朝后。
\end{comment}

%gsf：根据图例，简单介绍整个order的流程。就是上边说的1、2、3、4.
The overview of \go is illustrated in Fig. \ref{fig:go-eg}. It \eat{consists of 4 steps as follows.}comprises following steps, 1) \textit{extract high-degree \& isolated vertices}, 2) \textit{divide other vertices}, 3) \textit{reorder vertices within subgraphs}, 4) \textit{reorder subgraphs} and 5) \textit{insert high-degree \& isolated vertices.} % as outlined below. \eat{why don't you include these steps in an algorithm, maybe not needed, but a formal algorithm seems more professional}

Next, we will introduce the intuitions and details of each step.

\begin{figure}
\centering
  \begin{subfigure}{0.24\textwidth}
    \centering
    \includegraphics[width=0.65\linewidth]{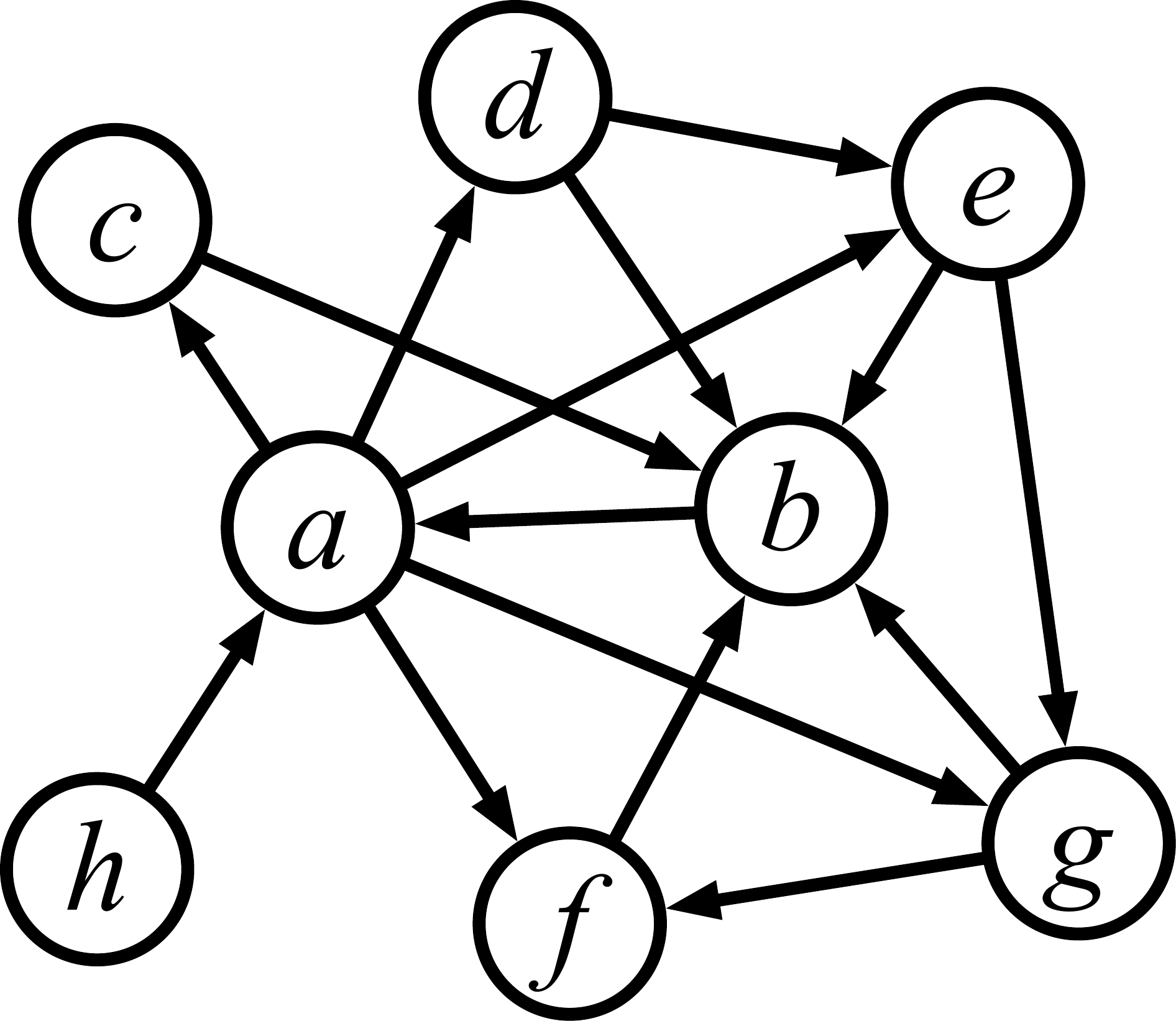}
    \caption{The initial graph.}
    \label{fig:sec3-a}
  \end{subfigure}%
  \begin{subfigure}{0.24\textwidth}
    \centering
    \includegraphics[width=0.65\linewidth]{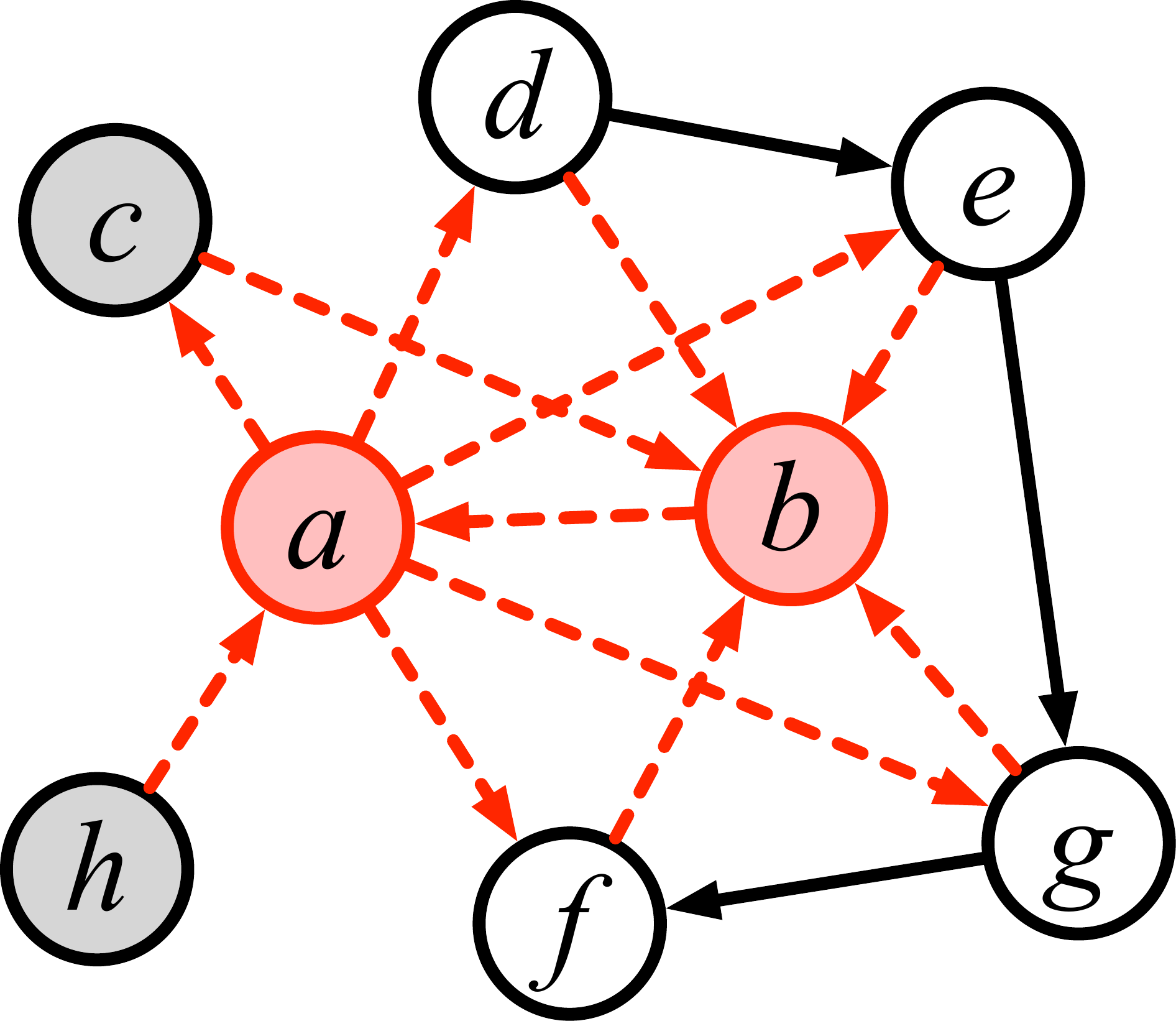}
    \caption{Extract high-degree vertices.}
    \label{fig:sec3-b}
  \end{subfigure}
  
  \vspace{\baselineskip} % Add some vertical space between the rows
  
  \begin{subfigure}{0.18\textwidth}
    \centering
    \includegraphics[width=0.8\linewidth]{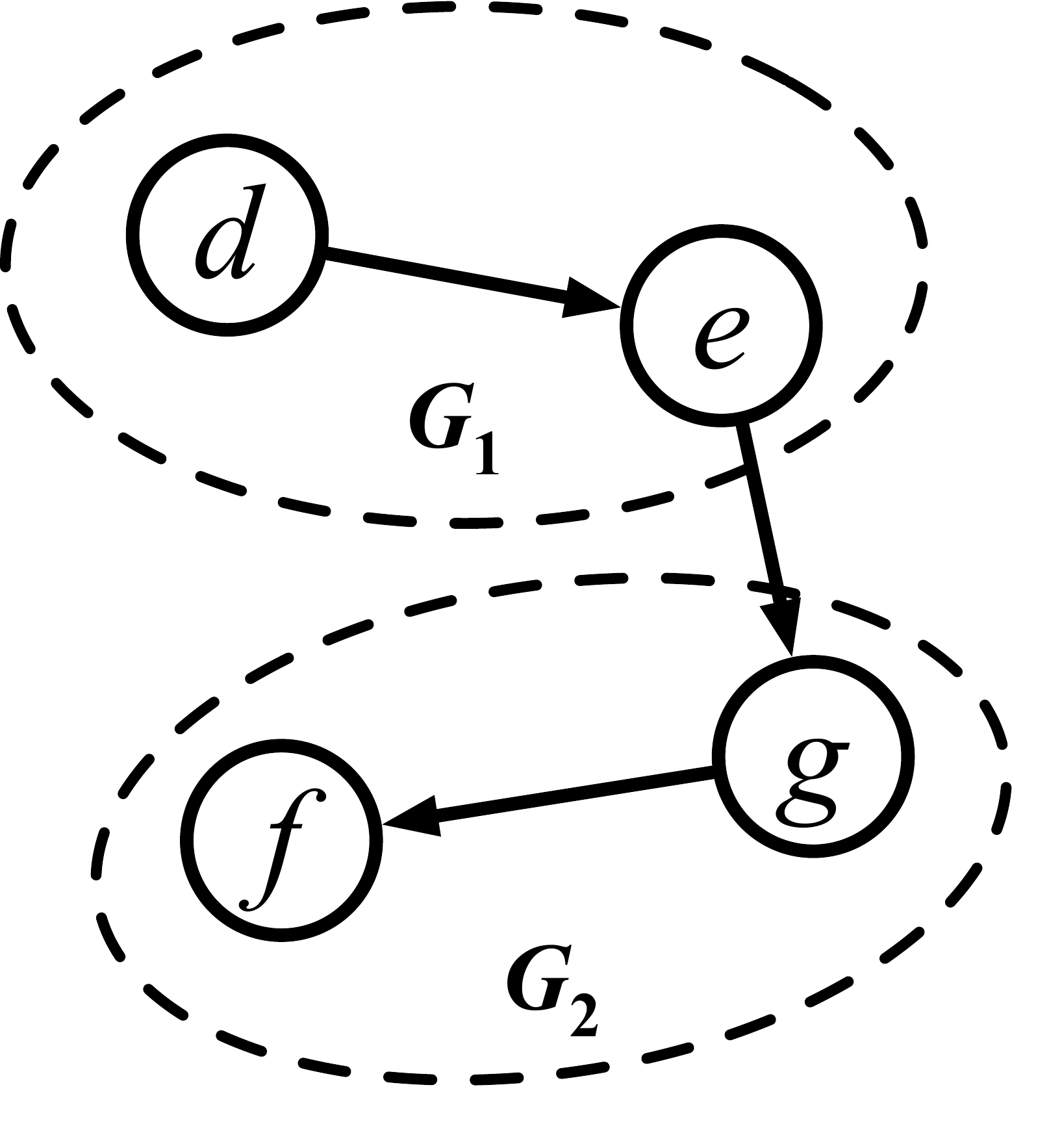}
    \caption{Divide the remaining vertices}
    \label{fig:sec3-c}
  \end{subfigure}%
  \hspace{0.2in}
  \begin{minipage}{0.5\linewidth}
  \vspace{-1.4in}
  \begin{subfigure}{1\textwidth}
    % \centering
    \includegraphics[width=0.75\linewidth]{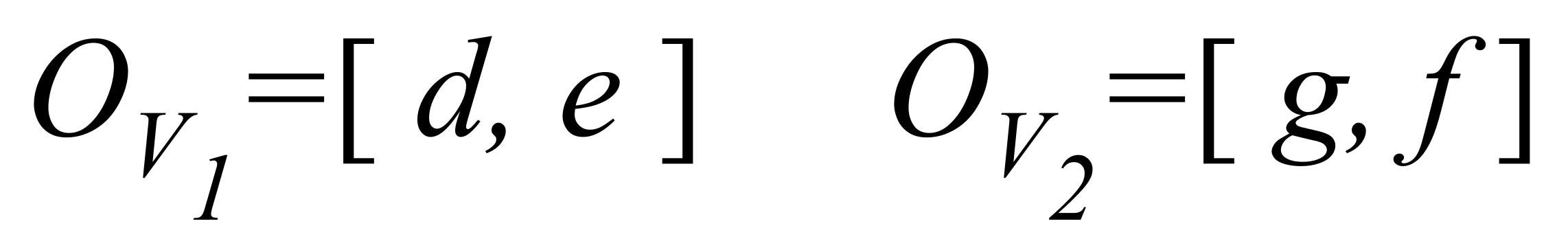}
    \caption{Reordering vertices intra-subgraphs.}
    \label{fig:sec3-d}
    \vspace{0.25in}
    \includegraphics[width=1\linewidth]{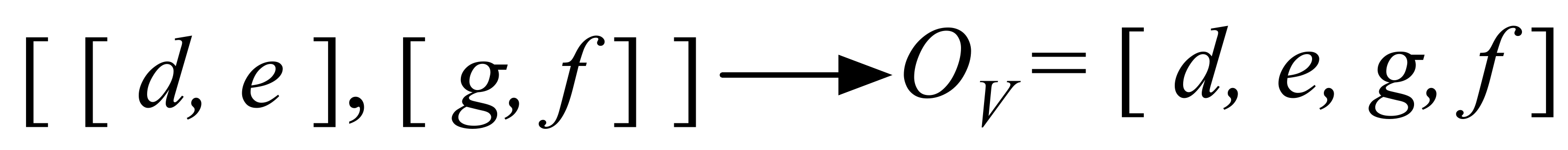}
    \caption{Reordering vertices inter-subgraphs.}
    \label{fig:sec3-e}
  \end{subfigure}
  \end{minipage}
  %\vspace{\baselineskip} % Add some vertical space between the rows
  
  \begin{subfigure}{0.48\textwidth}
    \centering
    \includegraphics[width=0.95\linewidth]{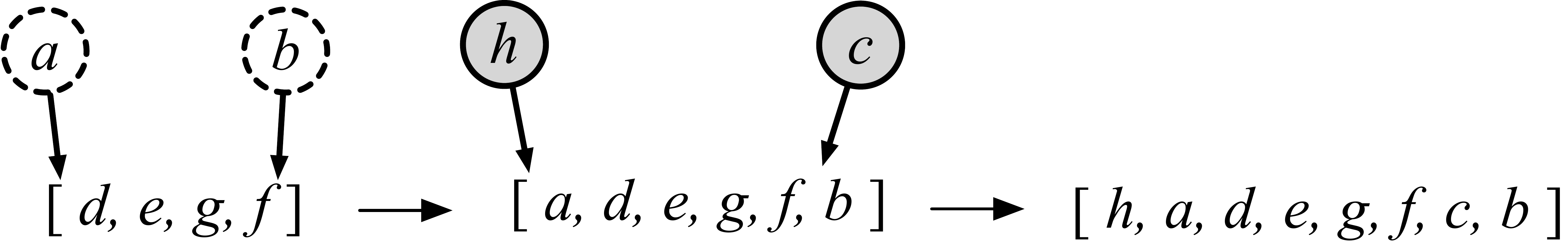}
    \caption{Inserting high-degreed/isolated vertices into the processing order.}
    \label{fig:sec3-f}
  \end{subfigure}
  \caption{An illustrative example of \go}
  \label{fig:go-eg}
\end{figure}

\etitle{Extract high-degree vertices}. It is widely acknowledged that most real-world graphs exhibit a power-law property, \eat{wherein a small number of vertices possess exceedingly high degrees, while a large proportion of vertices have low degrees.}where a very small number of vertices have extremely high degrees, while the majority of vertices have lower degrees. This property \eat{introduces a challenge during}{poses a challenge for} the vertex reordering process. %High-degree vertices establish connections with a substantial number of other vertices. 
In the process of reordering, the placement of high-degree vertices in the reordering subsequence significantly influences the positioning decisions of numerous lower-order vertices. Thus we first remove the high-degree vertice $V_{HD}$ and their edges $E_{HD}$ from the graph.
% When reordering vertices with lower degrees, they are \yaof{dramatically} influenced by the presence of high-degree vertices. 
\eat{This correlation not only increases the difficulty of the reordering but may also reduce the overall quality of the reordering results. This correlation not only increases the complexity of the reordering but may also reduce the overall quality of the reordering results.}
% As we all know, most graphs have power-law property, where few vertices have a very high degree and a large number of vertices have a very small degree. On the one hand, high-degree vertices connect a large number of vertices. When reordering the vertices with small degrees, they are affected by high-degree vertices, which increases the difficulty of reordering and even reduces the quality of reordering results. 
% \yaof{Moreover, these high-degree vertices naturally play a more important role in the iterative computation, since most vertex state changes pass through them\cite{zhang2016hotgraph}. For the reordering method in which they have a higher weight, a separate consideration of them is warranted.}

We take %the graph in 
Fig. \ref{fig:sec3-a} as an example to illustrate this point.  
%One approach is for the high-degree vertices to participate in the reordering process of the other vertices. 
If the high-degree vertices are not removed, 
%initially, we reorder the high-degree vertices $a$ and $b$, assigning a higher ordinal number to $a$ than $b$ to form a positive edge. 
we may first reorder $a$ and $b$ and assign a higher ordinal number to $a$ than $b$ to make edge $(b, a)$ positive.
Consequently, we obtain a temporary processing order $[b, a]$. Then we insert other vertices into $[b, a]$. In order to achieve the optimal vertex processing order, 
%i.e., to maximize the value of $\M(\cdot)$, 
%we adopt a naive reorder method that sequentially \emph{attempts inserting} the remaining vertices into the subsequence $[b, a]$ to obtain the current maximum $\M(\cdot)$.
when inserting each vertex, we try to maximize $\M(\cdot)$ value.
Ultimately, the optimal vertex processing order is $O^1_{V}=[d, e, c, b, h, a, g, f]$. 
On the contrary, 
%another approach is that we initially extract $a$ and $b$ from the graph and use the same naive reorder method for the remaining vertices, resulting in the subsequence $[h,c,d,e,g,f]$. 
after removing $a$ and $b$, we first reorder $d, e, f$, and $g$, where $h$ and $c$ become isolated and also be removed. It is evident that the optimal processing order is $[d,e,g,f]$.
Finally, we insert $a, b$ and $c, h$ into $[d,e,g,f]$ and obtain the processing order $O^2_V=[h, a, c, d, e, g, f, b]$. The $\M(\cdot)$ values %for these two processing orders are 
of these two orders are $\M(O^1_V)=10$ and $\M(O^2_V)=14$. \eat{Obviously}Based on the previous definition of $\M(\cdot)$, $O^2_V$ is a better processing order.
Although the edge $(b, a)$ is sacrificed in $O^2_V$ due to $a$'s higher ordinal number than $b$, it brings 
%the source's ordinal number of a large number of edges smaller than the destination's ordinal number.
more positive edges.

It is notable that after removing the high-degree vertices and their edges from the graph, there will appear some isolated vertices that have no edges with other vertices. As illustrated in Fig. \ref{fig:sec3-b}, vertex $c$ and $h$ become isolated vertices after removing $a, b$ since they only connect with $a$ and $b$. \eat{Because there are no edges between isolated vertices and other vertices, they have no effect on reordering other vertices. So we remove isolated vertices $V_{I}$ by the way.}Since isolated vertices have no edges connecting them to other vertices, they do not influence the reordering of other vertices. Therefore, we eliminate isolated vertices, denoted as $V_{I}$.

\etitle{Divide other vertices}. Due to the complex interconnection between vertices, although the complexity of the graph structure is reduced after removing high-degree vertices and isolated vertices, it is still difficult to reorder vertices from the perspective of the whole graph. Therefore, to simplify reordering, we divide the remaining graph into smaller subgraphs.
Then design the reordering method considering both intra- and inter-subgraphs perspectives.
% Then reordering the vertices by employing local and global reordering methods that reorder the vertices within and between partitions. 
%Graph partitioning is to simplify vertices reordering so that the reordering tasks on the whole graph are divided into smaller subgraphs. 
%In order to make better local reordering results within subgraphs bring better global reordering results. 
The graph dividing result requires that there are as many edges within the subgraph as possible and as few edges between the subgraphs as possible. The reason is explained as follows.

Firstly, when reordering vertices within a subgraph, if there are few edges between vertices intra-subgraph and many edges between vertices inter-subgraph, the better local reordering results may not bring a better global reordering result. Instead, the result of global reordering is determined by the reordering result of vertices between subgraphs. 
Secondly, \gr{\cite{wei2016gorder} points out that the locality of vertices in the \\ processing order has a great impact on CPU cache performance. In most graph algorithms, accessing a vertex often necessitates accessing its neighbors. For example, in PageRank, updating a vertex requires accessing the states of its incoming neighbors. Sequential accesses to the neighbor states of vertices stored in more distant locations in physical memory may result in serious cache misses. Therefore, vertices connecting each other should be as close as possible in the processing order.} Thus, we should group the vertices that are closely connected into the same subgraphs.

%Secondly, if there are many connected edges between partitions, the adjustment of the partition's rank causes changes in the rank value of the source and destination of a large number of edges, which increases the difficulty of partition reordering and decreases the quality of processing order. 

Therefore, when dividing the graph, there should be as many edges as possible within the subgraph and as few edges as possible between subgraphs. This is similar to the purpose of graph community detection or graph partitioning in distributed computing, so we can employ these methods such as Louvain \cite{louvain} or Metis \cite{metis}. Finally, we obtain a set of subgraphs of the graph $\{G_1, \cdots, G_K\}$, where $G_i=\{V_i, E_i\}$.

%图划分后我们可以在每个子图内部进行顶点的重排序。在子图内部，首先可以随机选择一个点作为初始点，通常我们可以选择入度最小且出度最大的点，因为这个点的入邻居较少往往排在order的前端。然后采用类似bfs的方法往subsequence中插入顶点，同样，bfs在遍历出邻居时首先处理出度最大的点，因为这些点排在前面的可能性较大。然后依次将邻居插入到现有的order中。在每插入一个点时，首先考虑的是插入到order的末尾，如果在末尾的位置其所关联的边不超过\alpha*edge，则调整其位置，使其插入一个使\M最大的位置，其中alpha大于1/2小于1是一个容忍度，可以容忍一定数量的不可以，因为采用bfs放在后边可以提升邻居的紧密程度，在迭代计算过程中提升cpu cache的命中率，如果单纯为了边的朝向而使cache命中率降低太多也不行。可以看出，没插入一个点，其本质是根据该点与order中已有的点之间的连边进行插入，使这些边的source的rank尽量大于destination的rank。【配fig2图说明】
\etitle{Reorder vertices within subgraphs}. After dividing the graph, we first reorder the vertices inside each subgraph $G_i$. To begin, we can randomly select a vertex as the initial vertex of the processing order $O_{V_i}$ formed by the vertices $V_i$ in the subgraph. In practice, the initial vertex always has the smallest in-degree, since such vertex tends to rank at the front of the processing order.  
% \yaof{Intuitively, such vertices tend to form fewer negative edges.}
We then select a vertex from the remaining vertices to insert into the processing order. We prefer selecting the vertex $v$ with BFS,  
%that has more connected edges with vertices that have been inserted into the \emph{current} processing order $O_{V_i}^c$. 
so that $v$ has better locality with the vertices in $O_{V_i}^c$.

For each selected vertex $v$, %the neighbors of $v$ that are in $O^c_{V_i}$ are $N^c_v = \{IN(v) \cup OUT(v)\} \cap O^c_{V_i}$, the edges between $v$ and $N^c_v$ are $E^c_v$. %edges connected with vertices in  $O_{V_i}^c$, \red{where $N^c_v \subseteq IN(v) \cup OUT(v)$ is the set of $v$'s neighbors  that have been inserted in $O^c_{V_i}$}. 
we search the optimal insertion position from the tail to the head of $O_{V_i}^c$ that maximizes the $\M(\R(G_i))$ value based on the position of the vertices that are in $O^c_{V_i}$. %edged with $v$ in $N^c_v$.
\eat{We first consider inserting it at the end of $O_{V_i}^c$, the positive edges are $E^c_v \cap IN(v)$, \ie, the edges from vertices in $O_c$ to $v$. %If $|E^c_v \cap IN(v)| \geq \alpha\times |E^{R_c}_v|$, where $1/2 \leq \alpha < 1$ is the tolerance to small $\M()$ value and $|E^{R_c}_v|$ is the number of edges connected $v$ and the vertices in the current processing order $R_c$,
%associated edges at the current position do not exceed edge $\alpha\times edge$, 
We then search the optimal inserting position from the tail of $O_c$ to the head that maximizes the $\M(\R(G_i))$ value.}
%adjust its rank to insert it in a position that maximizes $\M()$ of the subgraph. We prefer to insert the end of processing order to enhance the tightness of neighbors, thereby improving CPU cache hit rates during iterative computations \cite{wei2016gorder}. 
%It is not acceptable to significantly reduce cache hit rates just for the sake of edge orientation. 
Intuitively, the essence of inserting vertices is to maximize the number of %connected edges with the existing vertices whose rank of the source of the edges is smaller than the destination, meanwhile, togethering the neighboring vertices. 
positive edges.
%For example in Fig. \ref{fig:sec3-d}, in $p_0$, we select $d$ as the initial vertex $[d]$, then travel vertex $e$ along edge $(d, e)$, and insert $e$ into the end of $[d]$. Finally, we get $[d, e]$.
It is worth mentioning that the overhead of such sequential attempts is not substantial, since we only need to search at locations near $v$'s neighbors in $O_{V_i}^c$ (the details will be discussed in implementation, Section \ref{sec:go:impl})

\eat{Next, we reorder the vertices between different subgraphs, which is done by reordering subgraphs.}

\etitle{Reorder subgraphs}. Next, we reorder the subgraphs to merge the order of vertices within subgraphs into a whole processing order. 
\gr{To make the vertices in the same subgraph continuous in the processing order,  we treat the entire subgraph as a super vertex, 
% which makes the vertices in the same subgraph closer together in the memory, 
which brings vertices within the same subgraph physically closer in memory,
thereby reducing CPU cache misses.} There will be edges between two super vertexes if there are edges between vertices in these two subgraphs, as shown in Fig. \ref{fig:sec3-c}, there is an edge between $G_0$ and $G_1$ since there is an edge $(e,g)$. 
Then we reorder the super vertices formed by subgraphs using a similar method that reorders vertices within subgraphs.
The difference is that there is a weight value on the edge between subgraphs, which is defined as the number of edges from one subgraph to another one, i.e., $w_{G_i,G_j}=|\{(u,v)|u \in G_i, v \in G_j\}|$. Then the objective function of ordering super vertices becomes 
$\M(O_P)=\sum_{G_i, G_j \in P} \chi_c(G_i,G_j)$, where $\chi_c(G_i,G_j)=w_{G_i, G_j}$ if $\ord(G_i) < \ord(G_j)$, else $\chi_c(G_i,G_j)=0$, where $\ord(G_i)$ returns the ordinal number of super vertex formed by $G_i$. 

After reordering subgraphs, the processing order of $G'$ will be obtained by decompressing the super vertices formed by subgraphs.
%Therefore, We prefer to place the source partition of the edge with greater weight in front of the destination partition. 
%For example in Fig. \ref{fig:sec3-d}, the $G_0$ is in the front of $G_1$, and finally, we obtain the processing order $[d,e,g,f]$.
% \begin{equation*}
%     \M(O_P)=\sum_{G_i, G_j \in P} \chi_c(G_i,G_j),
% \end{equation*}
% where $\chi_c(G_i,G_j)=w_{G_i, G_j}$ if $\ord(G_i) < \ord(G_j)$, else $\chi_c(G_i,G_j)=0$

\etitle{Insert high-degree \& isolated vertices.}
%目前为止，我们已经将除高度点和孤立点意外的其它顶点重排序完毕，接下来将高度点和孤立点插入到现有的processing order中。由于孤立点与现在processing order中的点没有边，因此我们先将高度点插入processing order中，然后再插入孤立点。由于高度点往往与processing order中的很多点有边，因此找到一个非常好的插入位置，满足其所有边的source的rank值比destination的小并不是一件容易的事。但是我们还是能够使大部分的边满足该要求。首先对于每一个大度点v，如果出度比入度大，我们优先考虑将其插入到processing order的头部。If the number of edges whose source rank value is smaller than the destination is larger than $\alpha\times |E^{R_c}_v|$, 则从前往后寻找插入v的最佳位置，使\M值增加最大。反之如果v入度比出度大，则优先考虑将其插入到processing order的末尾，否则从后往前查找最佳位置。而对于孤立点，我们采用与大度点相同的方法插入到现有processing order中。例如图2e，我们优先将顶点a和h插入到现有processing order中的头部，将顶点b和c优先插入到processing order中的尾部。由于点c不满足大于等于$\alpha\times |E^{R_c}_v|$的要求，将c往前挪一下，发现已经最大，插入到b前面。
So far, we have reordered all vertices except for those with high degrees and isolated vertices. %Next, we will insert high-degree vertices and isolated vertices into the existing processing order. 
Since isolated vertices do not have connected edges with the vertices presently in the processing order, we prioritize inserting vertices with high degrees into the processing order, followed by the insertion of isolated vertices. %Considering that high-degree vertices typically have many connected edges with the vertices in the processing order, it is challenging to find an optimal insertion position that satisfies the condition where the rank values of the source vertices are smaller than those of the destination vertices for all connected edges. Nevertheless, we still strive to ensure that the majority of edges meet this requirement.
% Initially, for each high-degree vertex $v$, if its out-degree exceeds its in-degree, we prefer inserting it at the front of the current processing order to structure more positive edges. 
% If the number of positive edges is smaller than $\alpha\times |E^{R_c}_v|$, Subsequently, 
We search for the optimal position to insert $v$ from front to back, maximizing the increase in the $\M(\cdot)$ value. The insertion method is similar to the approach used when reordering vertices within a subgraph.
% Conversely, if $v$'s in-degree is greater than its out-degree, priority should be given to inserting it at the tail of the processing order, and the best position should be sought from back to front. 
For isolated vertices, we employ the same method as inserting high-degree vertices into the existing processing order. 
% For example, in Fig. \ref{fig:sec3-e}, we first insert vertices $a$ and $h$ at the beginning of the existing processing order, then place vertices $b$ and $c$ at the end of the processing order. As vertex $c$ does not meet the requirement of being greater than or equal to $\alpha\times |E^{R_c}_v|$, it is moved forward, and upon finding it already at its maximum, it is inserted in front of $b$.

Next, we analyze the effectiveness and efficiency of \go theoretically.

\subsection{Effectiveness}
%[]已经证明了lower bound是1/2，那么我们的方法能否达到这个要求呢？根据下面的分析我们的方法得到的结果的下界也是1/2.
%Does the processing order obtained by this method have a higher $\M()$ value than a random processing order? Randomly, the probability that the rank of the source of each edge is greater than the rank value of the destination is 1/2. The rank value of the source of half of the edges is greater than the value of the destination, so the value of $\M()$ is $|E|/2$. 
%It has been proved that the theoretical lower bound of positive edges is $|E|/2$ \cite{hassin1994approximations,cvetkovic2020maximal}.
Randomly, the probability of each edge being positive edge is 1/2, and the $\M(\R(G))$ is $|E|/2$. According to the following analysis, it can be seen that the $\M(\R(G))$ value of the processing order obtained by \go is no smaller than $|E|/2$, which is no worse than a random processing order.

\begin{figure}[t] % h为当前位置，!htb为忽略美学标准，htbp为浮动图形
\centering
\includegraphics[width=0.6\linewidth]{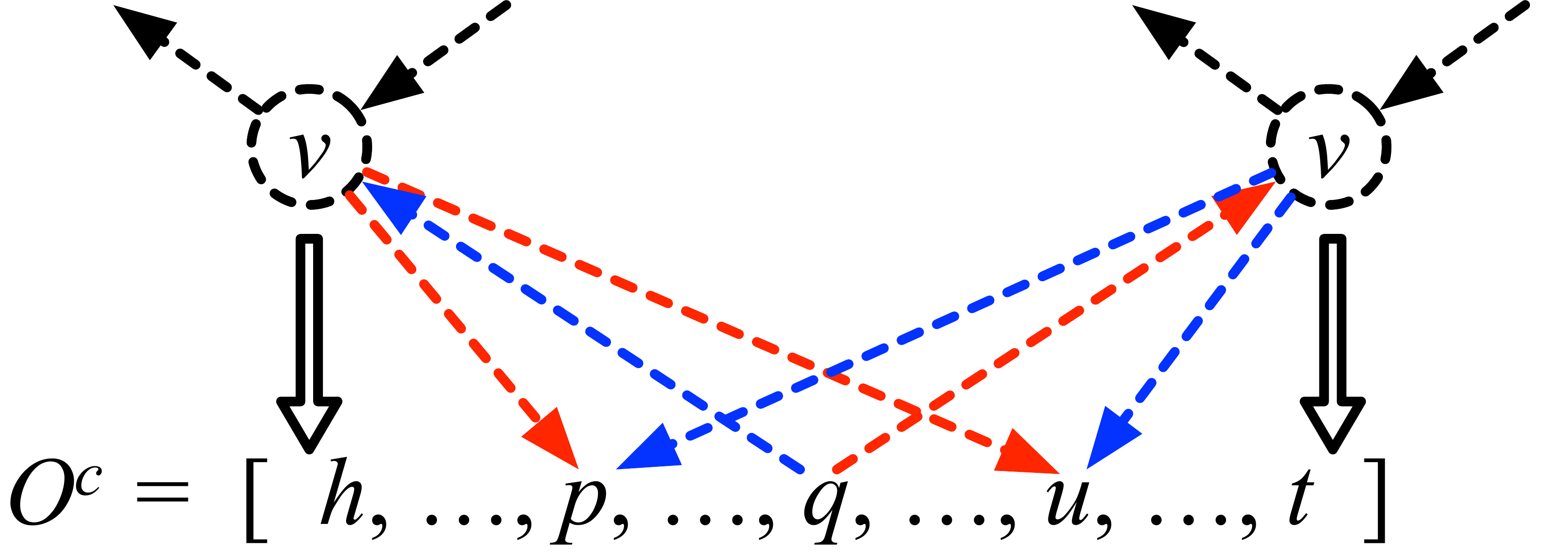} 
\caption{Example of positive and negative edges when vertex $v$ is inserted into the head or tail of $O^c$, where red dashed lines represent positive edges and blue dashed lines represent negative edges}\label{fig:theo-2}
\end{figure}

Before the theoretical analysis, we first propose the following theorem.

\begin{lemma}\label{lem:posnum}
    Given a current processing order $O^c$, and a candidate vertex $v$ to be inserted into $O^c$, after inserting $v$ into the $O^c$ using \go, the value of $\M(O^c)$ will increase by at least $|E^c_v|/2$, where $|E^c_v|$ is the number of edges connecting $v$ and vertices in $O^c$.
\end{lemma}
\begin{proof}
    There are two situations after $v$ is inserted into $O^c$, \romannumeral1) $v$ is inserted into the head or tail of $O^c$, and \romannumeral2) $v$ is inserted into another position of $O^c$. 

    We denote the neighbors of $v$ that are in $O^c$ are $N^c_v = \{IN(v) \cup OUT(v)\} \cap O^c$, then the edges between $v$ and $N^c_v$ are $E^c_v$, \ie $|N^c_v|=|E^c_v|$.
    
    For the first case, when $v$ is inserted into the head (tail) of $O^c$, outgoing edges (incoming edges) of $v$ in $E^c_v$ are positive edges since $v$'s ordinal number is smaller (larger) than its outgoing neighbors. Then, the incoming edges (outgoing edges) are negative.
    %are positive edges greater, and incoming edges are negative edges. If $v$ is inserted into the tail of $R_c$, $v$'s incoming edges become positive edges, and outgoing edges become negative edges. For example in Fig. \ref{fig:theo-2}, if $v$
    Thus, when $v$ is inserted into head or tail of $O^c$, there are $max(|\{N^c_v \cap IN(v)\}|, |\{N^c_v \cap OUT(v)\}|)$ edges are positive edges, where $N^c_v \cap IN(v)$ and $N^c_v \cap OUT(v)$ are incoming neighbor and outgoing neighbors of $v$ in $O^c$ respectively. As we all know $max(|\{N^c_v \cap IN(v)\}|, |\{N^c_v \cap OUT(v)\}|) \geq (|\{N^c_v \cap IN(v)\}|+|\{N^c_v \cap OUT(v)\}|)/2=|N^c_v|/2$. Therefore, the number of positive edges is no smaller than $|E^c_v|/2$.

    For the second case, $v$ is inserted into another position that maximizes the $\M(\cdot)$ value. It means that we find a better inserting position that makes the number of positive edges larger than inserting into the tail or head. Then the number of positive edges is larger than $|E^c_v|/2$.

    In summary, after inserting $v$ into $O^c$, at least $|E^c_v|/2$ edges are positive edges.
\end{proof}

For the example in Fig. \ref{fig:theo-2}, $(v,p), (q,v), (v,u)$ are edges of $v$ that connect $v$ and vertices in $O^c$. %, \red{$N^c_v\subseteq IN(v) \cup OUT(v)$}. 
If $v$ is inserted into the head of $O^c$, $v$'s outgoing edges $\{(v,p),(v,u)\}$ are positive edges while incoming edge $(q,v)$ is a negative edge. If $v$ is inserted into the tail, then $(q,v)$ becomes a positive edge, while $(v,p)$ and $(v,u)$ become negative edges.

%When inserting each vertex $v$ into the current processing order $R_c$, the essence is to divide the edges connecting $v$ with the vertices in $R_c$ into positive or negative edges. And once the edge is determined to be positive or negative, it will not change, since if $r(u) < r(v)$ after inserting $v$, then $r(u)$ is always smaller than $r(v)$ regardless of partition reordering or insertion of other vertices. 

Based on the lemma, we have the following theorem.

\begin{theorem}
   After reordering vertices in graph $G(V, E)$, we obtain the processing order $O_V$. Then we have 
   \begin{equation}
       \M(O_V) \geq |E|/2
   \end{equation}
\end{theorem}
\begin{proof}
    During the reordering process, there are 4 types of edges, 1) the edges within each subgraph $E_{V_i}$, 2) the edges between subgraphs $E_P$, 3) the edges between high-degree vertices and the vertices in subgraphs $E_{HD}$, 4) the edges between isolated vertices and high-degree vertices $E_{ISO}$, \ie, $E=\bigcup^N_{i=1}E_{V_i}\cup E_P \cup E_{HD} \cup E_{ISO}$.

    In each subgraph, when inserting each vertex $v$, $|E^c_v|$ only contains the edges connecting $v$ and the vertices in the current local processing order $O^c_{V_i}$, not the edges connecting $v$ and other vertices. Therefore, for each edge $(u,v) \in E_i$, $(u,v)\in E^c_u$ if $v$ is inserted into $O^c_{V_i}$ before $u$ otherwise $(u,v)\in E^c_v$. Thus we have $E_{V_i}=\bigcup_{v\in V_i} N^c_v$. According to Lemma \ref{lem:posnum}, we have $\M(O_{V_i})\geq \sum_{v\in V_i}|E^c_v|/2=|E_i|/2$. 
    
    Similarly, for the other 3 types of edges, the number of positive edges between subgraphs is at least $|E_P|/2$, the number of positive edges between high-degree vertices and the vertices in subgraphs is at least $|E_{HD}|/2$, the number of positive edges between isolated vertices and high-degree vertices is at least $|E_{ISO}|/2$.

    Finally, we have the number of positive edges is at least $\sum^N_{i=1}|E_i|/2 + |E_P|/2 + E_{HD}/2 + |E_{ISO}|/2=|E|/2$, \ie, $\M(O_V) \geq |E|/2$.
    %在排序过程中，我们将边数据分成4种边，1）partition内的边，2）partition之间的边，3）大度点的边和4）孤立点的边。在每个partition内部，根据lemma，我们可以保证有一半以上的边为positive边。对于分组之间的边，同样根据lemma，不同的是每条边有权重，可以保证partition在reorder之后，分组之间的边的positive边的权重之和大于所有边的权重之和的一半，既positive边的数量大于一半。根据lemma，可以得知在每次插入大度点和孤立点时都能保证每点与当前order中的点的连边的一半以上都为postive，最后，得到的order的positive的总数大于一半。
\end{proof}

%In each partition, there at least $|E_i|/2$ edges are positive edges. Because there are at least $|E^c_v|/2$ edges becoming positive edges according to Theorem \ref{thm:posnum} when inserting each vertex $v$. Similarly, no less than half of the edges between partitions become positive edges. Furthermore, according to Theorem \ref{thm:posnum}, when inserting each high-degree vertex and isolated vertex, no less than half of the edges connecting it and vertices in the current processing order are positive edges. Finally, we have that no less than half of the edges are positive edges, i.e., $\M(\cdot) \geq |E|/2$ after reordering vertices using \go.   
%在顶点插入的时候实际是边数据的插入，既E^c_v的边插入，这时候有5种，一种是在每个partition内部，2、partition之间的边，2、大度点插入时的边，孤立点插入时的边。在每个partition内部，其排序过程优先插入结尾，如果不大于则往内部插入，根据定理的，每次插入的边都会有一半以上是postive的。在partition之间进行reorder时，其本质是cluster之间的边的插入，由于其插入方法与partition内部插入方式相同，也能保证一半以上的边是postive。在插入大度顶点与孤立顶点时，根据定理可以同样会保证一半以上的新加入的边是postive，综合在所有边的一半以上是postive edges。因此go算法能保证一半以上的边是postive的。
%算法的证明，证明该方法得到的\M值>=|E|/2，也就是比随机的好。
\begin{comment}
    lemma：给定一个processing order，来一个点插入进去，能保证跟这个点所关联的新插入的边，一半以上朝后。
    proof：极端情况下，放在第一个和最后一个，第一个的值为出度-入度，最后的值为入度-出度，要么等于0，要么一正一负，放在正的地方就可以保证一半以上朝后。
\end{comment}
\begin{comment}
    theorem：该方法得到的\M值>=|E|/2
    proof：第5步能够保证在每个cluster内部\M值大于等于一半的边。因为每插入一个点，根据lemma保证所关联的现有边能有大于等于一半的边朝后。
    同时cluster之间也是这么排序，因此能够保证cluster之间的边有一半及其以上朝后。
    大度点和孤立点插入一样的插入方法，因此能够保证一半及其以上的边朝后。
    所有情况加起来，就是至少一半以上的边朝后。
\end{comment}

\subsection{Implementation}\label{sec:go:impl}

\SetKwFunction{func}{GetOptVal}
\SetKwProg{Fn}{function}{}{}
% \SetKwProg{Fn}{function}{\{}{\}}

\begin{algorithm} 
    \caption{\go Algorithm Sketch}\label{algo:go}
    \KwIn {Graph \(G(V, E)\)}
    \KwOut{Vertex processing order $O_V$}
    % \For{each $v$}{
    %     $v.val=Float.MAX$;\label{glao:go:val} \tcp{v.val maintains the relative size of the ordinal number of $v$, a larger val means larger ordinal number}
    % }
    Init $v.val$ of each $v$ with $\infty$;\label{glao:go:val} \tcp{$val$ represents the value of the ordinal number%, a larger val means larger ordinal number
    } 
    % \tcp{\blue{***Divide phase***}}
    \tcp{\textbf{***Divide phase***}}
    Extract high-degree vertices, isolated vertices and their edges $G_{HD}$, $G_{ISO}$ from $G$, the remain vertices and edges form subgraph $G'$; \label{algo:go:extra}

    %$V = V_{HD} \cup V_{ISO} \cup V_{RM}$; 
    %\tcp{$V_{HD}$ are high-degree vertices, $V_{ISO}$ are isolated vertices, $V_{RM}$ are remaining vertices.}

    %$E = E_{HD} \cup E_{RM}$; 
    %\tcp{$E_{HD}$ are edges connected with $V_{HD}$, $E_{RM}$ are remaining edges.}

    % \tcp{graph partition}
    %$G_{RM} \leftarrow V_{RM} \cup E_{RM}$\;

    Divided graph $G'$ into $K$ subgraphs $\{G_1, \cdots, G_K\}$ with graph partitioning or clustering method;\label{algo:go:part}

    % \tcp{\blue{***Conquer phase***}}
    \tcp{\textbf{***Conquer phase***}}
    \For{each subgraph $G_i(V_i, E_i)$}
    {\label{algo:go:intra:start}
        %$Q_{tmp} \gets BFS(G_i)$; \tcp{init the candidate reorder }
        $O_{V_i} \gets \emptyset$; \tcp{init the processing order of $G_i$} 
        % select the smallest indegree vertex $s$; \\        
        % $Q_{tmp}.insert(s)$;\\        
        %\While{$Q_{tmp} \neq \emptyset$}{
        \For{each vertex $v$ in $V_i$}{
            % $v=Q_{tmp}.pop()$;\\
            % $N^c_v=\{IN(v) \cup OUT(v)\} \cap O_{V_i}$; \tcp{obtain v's neighbors that have been inserted into $O_{V_i}$}
            $v.val\gets$\func{$O_{V_i}$, $v$}\\
            $O_{V_i}.add(v)$;\\
            %$Q_{tmp}.insert(OUT(v)\cap V_i \setminus O_{V_i})$; \tcp{insert v's outgoing neighbors that are in $V_i$ but have not been visited}
        }
        % $O_{V_i} \leftarrow O_{V_i} \cup s$, $V_i \leftarrow V_i \setminus s$ \;
        % \For{each vertex $v$ in $V_i$}{
        %     $O_{V_i} \leftarrow$ \func{$O_{V_i}$, $v$} \;
        % }
    }\label{algo:go:intra:end}
    %\tcp{Form the super vertices with $\{G_1 \cdots G_K\}$}
    % \tcp{\blue{***Combine phase1: reorder subgraphs***}}
    \tcp{\textbf{***Combine phase: reorder subgraphs***}}
    \For{each subgraph $G_i$}{\label{algo:go:svs}
        create super vertex $sv_i$ with $V_i$;\\
        %$sv_i.V\gets V_i$;\\ %\tcp{create super vertex $sv_i$ with $V_i$}
        $sv_i.val \gets \infty$;\label{glao:go:sval} \tcp{$sv.val$ represents the value of the ordinal number of $sv$%, a larger val means larger ordinal number
    }%$sv_i.val=Float.MAX$;
    }\label{algo:go:sve}
    % \For{each edge $(u,v)$ where $u\in sv_i.V, v\in sv_j.V$}{\label{algo:go:ses}
    %     $w(sv_i, sv_j)++$; \tcp{compute the weight of edge $(sv_i, sv_j)$}
    % }\label{algo:go:see}
    \For{each $sv_i$ }{
        \For{each $sv_j$}{
        \tcp{compute the weight of edge}
            $w(sv_i, sv_j)\gets |\{(u,v)|u \in G_i, v \in G_j\}|$\label{algo:go:se}%\tcp{compute the weight of edge}
        }
    }
    Construct a graph $G_{sv}$ with vertices set $\{sv_i | 1<i\leq K\}$ and edges set $\{(sv_i, sv_j) | 1\leq i\leq K, 1\leq j\leq K\}$;\label{algo:go:gsv}\\ %$G_{sv}(\{sv_i\}, \{(sv_i, sv_j)\})$ \tcp{construct graph $G_{sv}$}
    %\tcp{inter-subgraph reorder}
    %$Q_{tmp}\gets BFS(G_{sv})$;\label{algo:go:rosvst}\\
    $O_{S}\gets \emptyset$;\tcp{init the processing order of $G_{sv}$}  
    % select the smallest indegree $vs_i$;\\
            
    % \While{$Q_{tmp} \neq \emptyset$}{
    %     $vs=Q_{tmp}.pop()$;\\
    %     $vs.val=$\func{$O_{P}$, $vs$}\\
    %     $O_{P}.add(vs)$;\\
    %     % $Q_{tmp}.insert(OUT(vs) \setminus O_{P})$; \tcp{insert vs's outgoing neighbors that have not been visited}
    % }\label{algo:go:rosven}
    \For{each super vertex $sv$}{\label{algo:go:rosvst}
        $sv.val\gets$\func{$O_{S}$, $sv$}\\
        $O_{S}.add(sv)$;\\
        % $Q_{tmp}.insert(OUT(vs) \setminus O_{P})$; \tcp{insert vs's outgoing neighbors that have not been visited}
    }\label{algo:go:rosven}
    % Sort $O_S$ ascending with $vs.val$;\label{algo:go:sortsv}\\
    $O_V\gets \emptyset$; \tcp{init the processing reorder of $G$}
    % \tcp{\blue{***Combine phase2: Unzip super vertices***}}
    \tcp{\textbf{***decompose super vertices, update $val$***}}
    Sort $O_S$ in the ascending order of $vs.val$;\label{algo:go:sortsv}\\
    $mval_{pre}\gets 0$;\label{algo:go:unsvst}\\ %\tcp{maintain the max $val$ of vertices in current super vertex and previous super vertex respectively}
    \For{each super vertex $sv \in O_S[i]$}{ 
        $mval_{cur}\gets 0$;\\
        %$V_{tmp}=O_P[i].V$\\
        %sort vertices in $V_{tmp}$ ascending with $v.val$;\\
        \For{each vertex $v$ in $sv$}
        {
            $v.val\gets v.val+mval_{pre}$;\\%\tcp{update the local $val$ within subgraph to global $val$}
            $mval_{cur}\gets \text{max}(mval_{cur}, v.val)$;\\
            $O_V.add(v)$;
        }
        $mval_{pre}=mval_{cur}$;
    }\label{algo:go:unsven}
    % \tcp{\blue{***Combine phase3: reorder high-degree \& isolated vertices***}}
    \tcp{\textbf{***reorder high-degree \& isolated vertices***}}
    \For{each vertex $v$ in $V_{HD}$}
    {
        $v.val \gets  $ \func{$O_V$, $v$};\\
        $O_V.add(v)$;
    }
    % \tcp{insert isolated vertices}
    \For{each vertex $v$ in $V_{ISO}$}
    {
        $v.val \gets  $ \func{$O_V$, $v$};\\
        $O_V.add(v)$;
    }
    
    % Sort $O_V$ ascending with $v.val$;\label{algo:go:finsort}\tcp{obtain the final processing order}
    Sort $O_V$ in the ascending order of $v.val$;\label{algo:go:finsort}\tcp{obtain the final processing order}
    % \func{a vertices order $O^c$, a vertex $v$}{
    %     \tcp{insret $v$ into $O^c$ to get the maximun $\M$}
    %     abc111\;
    %     abc222\;
    % }
    
    % \Fn{\func{order $O$, vertex $v$}}{
    % %从后往前遍历，遇到能够使M值最大的位置，就是最佳插入位置
    %     $Neibs \leftarrow OUT(v) \cap O$\;
    %     $currM, maxM \leftarrow |Neibs|, bestOrd \leftarrow 0$ \;
    %     \For{$v'$ in $Neibs$}
    %     {
    %         \eIf{$v' \in OUT(v)$}
    %         {
    %             $currM \leftarrow currM - 1$\;
    %         }
    %         {
    %             $currM \leftarrow currM + 1$\;
    %             \If{$currM > maxM$}
    %             {
    %                 $maxM \leftarrow currM, bestOrd \leftarrow v'$
    %             }
    %         }
    %     }
    %     $O \leftarrow v$ inserts after $bestOrd$\;
    % }
\end{algorithm}

\setlength{\algotitleheightrule}{0pt}
\begin{algorithm}
% \SetKwFunction{IncSQ}{IncSQ}
% \SetKwProg{Fn}{Function}{:}{}
\Fn{\func{order $O$, vertex $v$}}{
        $N_v \gets\{IN(v)\cup OUT(v)\} \cap O$; \label{fun:neib}\\
        % \tcp{$N_v$ contains the incoming and outgoing neighbors of $v$ that have been inserted into $O$}
        % sort vertices in $N_v$ ascending with $v.val$;
        Sort $N_v$ in the ascending order of $v.val$;\label{fun:sort}\\
        % $val=0, max_{pe}=-Float.MAX$;\\
        % $pe=|OUT(v)\cap O|$; \label{fun:peinit}\tcp{number of positive edges}
        $pe_v=|OUT(v)\cap O|$; \label{fun:peinit}\tcp{init the number of positive edges}
        $val\gets 0$; \tcp{init the $val$}
        $max_{pe_v}\gets -\infty$; \tcp{the max number of postive edges}
        % $currM, maxM \leftarrow |Neibs|, bestOrd \leftarrow 0$ \;
        % \For{$i=1; i<|N_v|; i++$}
        \For{each vertex $N_v[i]$ in $N_v$}
        {\label{fun:upvalst}
            \eIf{$N_v[i] \in OUT(v)$}
            {\label{fun:pesubs}
                \eIf{$v$ is super vertex}{
                    % $pe -= w(v, N_v[i])$;\label{fun:wesub}
                    $pe_v \gets pe_v - w(v, N_v[i])$;\label{fun:wesub}
                }{
                    % $pe -= 1$;
                    $pe_v \gets pe_v - 1$;
                }\label{fun:pesube}
            }
            {\label{fun:peadds}
                \eIf{$v$ is super vertex}{
                    % $pe += w(v, N_v[i])$;\label{fun:weadd}
                    $pe_v \gets pe_v + w(v, N_v[i])$;\label{fun:weadd}
                }{
                    % $pe += 1$;
                    $pe_v \gets pe_v + 1$;
                }
            }\label{fun:peadde}
            \If{$max_{pe_v} < pe_v$}
                {
                    $max_{pe_v} \gets pe_v$;\\
                    $val \gets (N_v[i].val + N_v[i+1].val)/2$; \label{fun:val} \tcp{insert $v$ between $N_v[i]$ and $N_v[i+1]$}
                }
        }\label{fun:upvaled}
    \Return{$val$};
}
\end{algorithm}
%\setlength{\algotitleheightrule}{0.8pt}

%强调两点，插入非常耗时，即使使用大小堆等。因此我们用一个value，val来近似代替ordinal value，val越大，ordinal num越大，根据value对顶点排序，得到一个整体的processing order，那么go的实现就是计算一个顶点的value，根据邻居的value。

\go adopts a divide-and-conquer idea, and a sketch of its implementation is shown in Algorithm \ref{algo:go}.
% A sketch of the implementation of \go is shown in Algorithm \ref{algo:go}.

% From the algorithm process of \go,
In terms of the core idea of \go, the essence lies in finding an optimal position of the processing order for vertex $v$ that maximizes $\M(\cdot)$, then inserting the candidate $v$ into this position. However, before obtaining the global processing order, it is impractical to obtain the real ordinal number of each vertex, because with the vertex insertion, the real ordinal number value of the vertex will continue to change. Therefore, we use $val$ to represent the ordinal number (line \ref{glao:go:val}). The larger the $val$, the larger the oral number of the vertex in the processing order. %, which can be implemented by assigning $v$ the average of the $val$ of the predecessor and the successor in the current processing order (line \ref{fun:val} in GetOptVal function).  
Finally, we sort vertices according to $val$ to derive the real processing order and the ordinal number of each vertex (line \ref{algo:go:finsort}).

\gr{Before computing the $val$ of vertices, we first extract the high-degree vertices (line \ref{algo:go:extra}). As a rule of thumb, we simply extract the top 0.2\% vertices with the highest degree. Then we divide the graph (line \ref{algo:go:part}) with the exiting graph partitioning method, such as Rabbit-Partition \cite{arai2016rabbit}, Metis \cite{metis}, and Louvain \cite{louvain}. In our implementation, we use the graph partitioning method introduced in Rabbit.}
 %Following the extraction of high-degree \& isolated vertices (line \ref{algo:go:extra}) and dividing the graph (line \ref{algo:go:part}).
We proceed to compute the $val$ of each vertex within $G_i$ (line \ref{algo:go:intra:start}-\ref{algo:go:intra:end}). It is notable that the vertex $val$ is a local value within $G_i$. After computing the local vertex $val$, we treat each subgraph as a super-vertex (line \ref{algo:go:svs}-\ref{algo:go:sve}), the edges between super-vertices have weights that are equal to the number of edges between the subgraphs (line \ref{algo:go:se}). Then we compute the $val$ of each super vertex (line \ref{algo:go:rosvst}-\ref{algo:go:rosven}). %decompress each super vertex to get the global $val$ of each vertex. 
After that, we sort the subgraphs ascending with $val$ of the super vertex (line \ref{algo:go:sortsv}). Then, we unzip the super vertices and obtain the global $val$ of each vertex in $G'$, which is done by adding the maximum $val$ of vertices in the previous subgraph to the $val$ of each vertex for $G_i$ (line \ref{algo:go:unsvst}-\ref{algo:go:unsven}). %adding the $val$ of each vertex is added to the maximum $val$ of the previous subgraph, thereby obtaining the global $val$ of all vertices in $G'$. 

%val的计算是通过寻找在OV中的最佳位置，然后取最佳位置前后邻居的val的平均值
The $val$ is computed by finding the current optimal position in the current processing order, and taking the average of the $val$ of the predecessor and successor to indicate a value in between. 
During the search for the optimal position, it is unnecessary to recompute the value of $\M(\cdot)$ for every potential insertion position of vertex $v$.
% However, when looking for the optimal position, it is not necessary to recompute the value of $\M(\cdot)$ for every position where the vertex $v$ can be inserted. 
%It suffices to count the number of positive edges $N^p_{v}$ when inserting vertex $v$ into the front of or behind $v$'s each neighbor. 
% It is enough to compute the value of $\M$ when the vertex $v$ is inserted into the position in front or behind each neighbor.
This is due to the fact that for a vertex not connected to $v$, the count of positive and negative edges remains unchanged when $v$ is inserted in front of or behind it. 
Consequently, the value of $\M(\cdot)$ also remains unchanged. As shown in Fig. \ref{fig:theo-2}, inserting $v$ into the front or behind of $h$, the value of $\M(O^c)$ remains unchanged since $(v,p)$ and $(v,u)$ are always positive.
% Because for a vertex that is not $v$'s neighbor, the number of positive and negative edges will not change when $v$ is inserted into the position in front or behind it, so $\M$ will also not change. 
Therefore, to find the best position in the current order, it is enough to count the number of positive edges $pe_v$ when $v$ is inserted into the front or behind each neighbor, which can be implemented as follows (shown in \eat{\textbf{GetOptVal}}\func function). 

Firstly, we extract $v$'s neighbors that are in the processing order and form a neighbor sequence (line \ref{fun:neib}), then we sort the neighbors of vertex $v$ in ascending order according to their $val$ and form a sorted neighbor sequence (line \ref{fun:sort}). %The number of neighbors of $v$ is $|IN(v)|+|OUT(v)|$, therefore the sorting time complexity is $O((|IN(v)|+|OUT(v)|)\cdot log(|IN(v)|+|OUT(v)|))$. 
Then, we traverse each position in the neighbor sequence from the head or tail to find the optimal position one by one, %, the time complexity is $O(|IN(v)|+|OUT(v)|)$. 
and update the number of positive edges $pe_{v}$ %be updated
incrementally (lines \ref{fun:upvalst}-\ref{fun:upvaled}), instead of recounting. If $v$ is inserted into the head of the neighbor sequence, the initial value of $pe_{v}=|OUT(v)|\cap O_V$ (line \ref{fun:peinit}). After moving the $v$ to the back of the next neighbor, $pe_{v}$ is updated. If the next neighbor is incoming neighbor, then $pe_{v}=pe_v+1$ (lines \ref{fun:pesubs}-\ref{fun:pesube}). Otherwise, the $pe_{v}=pe_v-1$ (lines \ref{fun:peadds}-\ref{fun:peadde}). Note that if $v$ is a super vertex, then the update granularity of $pe_v$ is the weight of connected edges (line \ref{fun:wesub} and line \ref{fun:weadd}).  When $v$ moves to the tail of the neighbor sequence, the value of $pe_{v}=|IN(v)|\cap O_V$. 
For the example in Fig. \ref{fig:theo-2}, we assume there is a neighbor sequence $[p,q,u]$. With $v$ as the head, initially the $pe_v=2$. After moving $v$ behind $p$, $pe_v=2-1=1$, since $p$ is an outgoing neighbor. After $v$ moves behind $q$, $pe_v=1+1=2$, since $q$ is an incoming neighbor.  
%Once the optimal position, i.e., the position of $v$ relative to which of its neighbors should be inserted, has been determined, we insert $v$ into the previous or next position of this neighbor in the processing order. 
Once the optimal position is determined, we compute the $val$ of $v$ with its predecessor and successor (line \ref{fun:val}). %insert $v$ into the corresponding position in the processing order.
% After finding the optimal position, i.e., we should insert $v$ into which neighbor's previous or next position, we insert $v$ into the previous or next position of this neighbor in the processing order. 

%It can be seen that the most complex step is neighbor sorting. On average, each vertex has $|E|/|V|$ neighbors, then the average time complexity of neighbor sorting is $O(|E|/|V|log(|E|/|V|))$. The overall time complexity of \go is $O(|V|\cdot |E|/|V|log(|E|/|V|))=O(|E|log(|E|))$.
%从gograph的算法过程看，执行方式是将一个candidate顶点插入到现有的顶点中，不管是partition内，还是partition间，还是大度点和孤立点。但是，在寻找最优位置时并不需要对v可以插入的每个位置都计算\M的值，只计算当顶点v的插入其邻居前面或后边的位置时的\M的值即可，因为对于不是v的邻居，插入其前面与后面的位置，positive和negative的边的数量不会发生变化，因此\M的值也不会发生变化。因此寻找最佳位置可以通过提取邻居顶点组成一个新的队列，然后根据oridinal number对邻居队列排序，最后从头或从尾部遍历邻居队列中的每一个位置从而找到最优位置，既应该插入到哪个邻居的前面或者哪个邻居的后面，然后将其插入到order中邻居的前面或者后边。从头到尾的过程并不需要没移动一次算一次分，而是通过对其初始值进行计算。如果在头部，则其初始值为OUT，每当遍历一个点，当该点为入邻居时，则该值+1，如果该点为出邻居，则该值-1。直到末尾，变为IN。因为加的总数为IN，间的总数为OUT。这也与FigX中的例子一致。总共有E条边，V个点，每个点平均有E/V条边，因此平均每个点的复杂度为E/VlogE/V，一共有V个点，总复杂度为ELogE/V。在实际的操作中，其时间开销往往比较小，因为大部分点已经在patition内部排好序，其复杂度较小。
%效率高，其复杂度为O(ELogE)

\begin{comment}
\subsection{Remarks}
introduce some intuition optimization when reordering vertices. \go is time-consuming.
%就是intuition上的方法，包括score+bfs和简单score排序。还有啥吗？

\subsection{Distributed \go}
\end{comment}

\section{Experiments}\label{sec:expr}

\eat{will be rewritten, and redo the some expr 
\ys{ The number of iteration rounds was emphasized earlier, so do we need to show the experimental results of the iteration rounds?}}

\subsection{Experimental Setup}
In default, the experiments are performed on a Linux server with Intel Xeon Gold 6248R 3.00GHz %72-core 
CPU, 98 GB memory, and it runs on 64-bit Ubuntu 22.04 with compiler GCC 7.5. %The L1, L2, and L3 cache size of the machine is 64KB, 1MB, and 35.8MB, respectively. 

\etitle{Graph Workloads}.
%pagerank, php, pr, sssp, bfs, sswp
% 
We use four typical graph analysis algorithms in our experiments, including PageRank\cite{page1999pagerank}, 
% Penalized Hitting Probability (PHP)\cite{php}, Personalized PageRank (PPR)\cite{jeh2003scaling}, 
Single Source Shortest Path (SSSP), Breadth First Search (BFS), and Penalized Hitting Probability (PHP)\cite{php}.
% Single Source Widest Path (SSWP)\cite{pollack1960maximum}. 
When the difference between vertex state value in two consecutive iterations is less than $10^{-6}$, PageRank and PHP are considered to have achieved convergence. For SSSP and BFS, the algorithms are considered to have reached convergence when all vertex states are no longer changing. %To ensure the precision of the collected data, we execute these algorithms five times during the experiments and report the average total running time and cache statistics.

\begin{table}[h]%[b]
% \vspace{-0.2in}
    \caption{Datasets}
    %\vspace{-0.1in}
    \label{tab:data}
    \centering
    \footnotesize
    {\renewcommand{\arraystretch}{1.2}
    \setlength{\tabcolsep}{6pt} %colums
    % \begin{tabular}{|l| c |c| c |}
    \begin{tabular}{l c c c }
        \toprule
        % \toprule
        % \hline
        {\textbf{Dataset}} &
        {\textbf{Vertices}} &
        {\textbf{Edges}} &
        {\textbf{Abbreviation}} \\
         \midrule
        %  \hline
         Indochina~\cite{nr} & 11,358 & 49,138 & IC \\
         % \hline
         SK-2005~\cite{nr} & 121,422 & 36,7579 & SK \\
         % \hline
         Google~\cite{google} & 875,713 & 5,241,298 & GL \\
         % \hline
         Wiki-2009~\cite{nr} & 1,864,433 & 4,652,358 & WK\\
         % \hline
         Cit-Patents~\cite{cit-Patents} & 3,774,768 & 18,204,371 &  CP \\
         % \hline
         LiveJournal~\cite{nr} & 4,033,137 & 27,972,078 & LJ  \\
         % \hline?
         \bottomrule
    \end{tabular}
    }
    % \vspace{-0.2in}
\end{table}

\etitle{Datasets}.
Six real-world datasets are used in our experiments, including indochina-2004 \cite{nr}, sk-2005 \cite{nr}, Google \cite{google}, wikipedia-2009 \cite{nr}, cit-Patents \cite{cit-Patents} and soc-livejournal \cite{nr}. The details of each dataset are outlined in Table \ref{tab:data}.

\newcommand{\fgwth}{0.25}
\newcommand{\fghs}{-0.14in}
\begin{figure*}[t]
\vspace{-0.15in}
  \centering
  \includegraphics[width=0.5\textwidth]{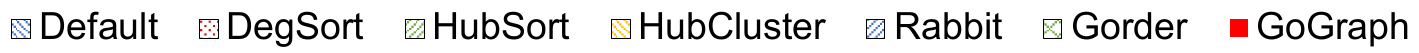}
  \\
  \begin{subfigure}{\fgwth\textwidth}
    \includegraphics[width=\linewidth]{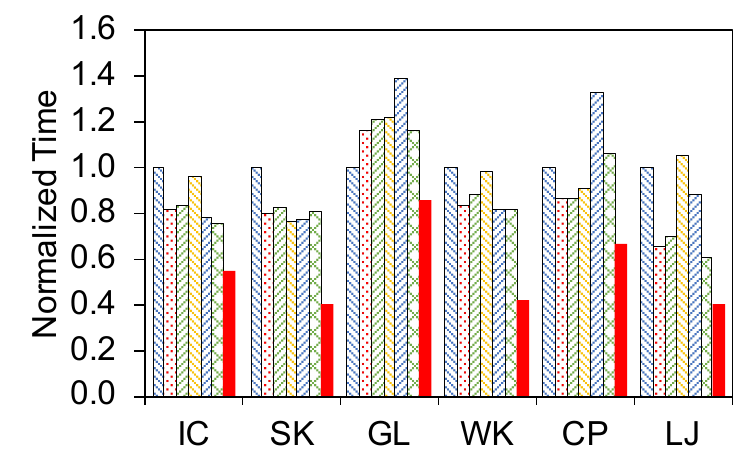}
    \caption{PageRank}
  \end{subfigure}
  \hspace{\fghs}
  \begin{subfigure}{\fgwth\textwidth}
    \includegraphics[width=\linewidth]{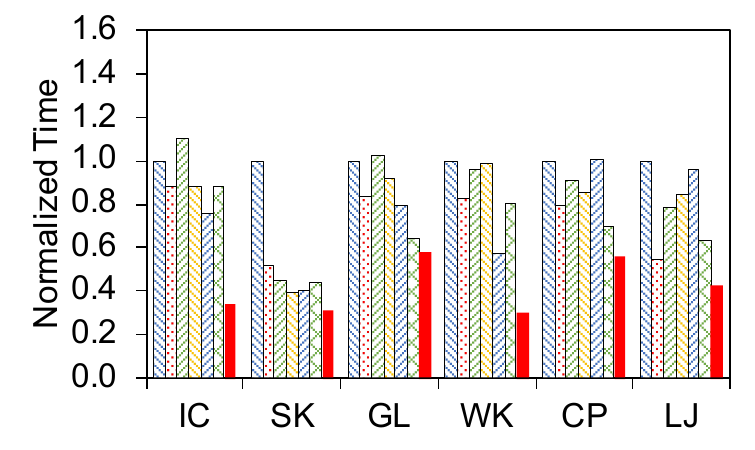}
    \caption{SSSP}
  \end{subfigure}
  \hspace{\fghs}
  \begin{subfigure}{\fgwth\textwidth}
    \includegraphics[width=\linewidth]{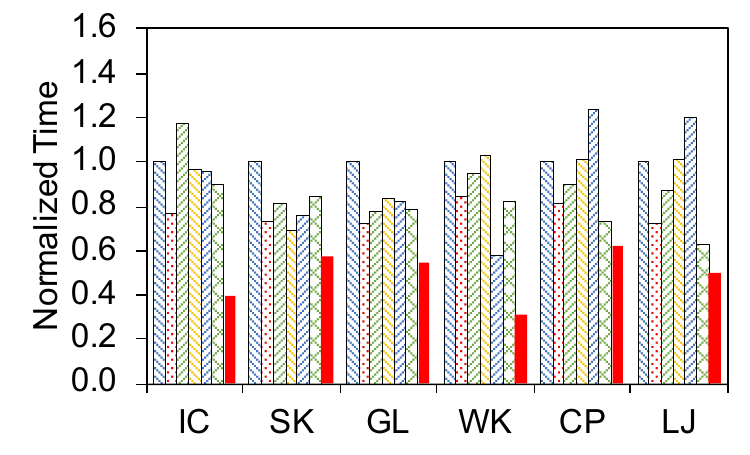}
    \caption{BFS}
  \end{subfigure}
  %\hfill
  \hspace{\fghs}
  \begin{subfigure}{\fgwth\textwidth}
    \includegraphics[width=\linewidth]{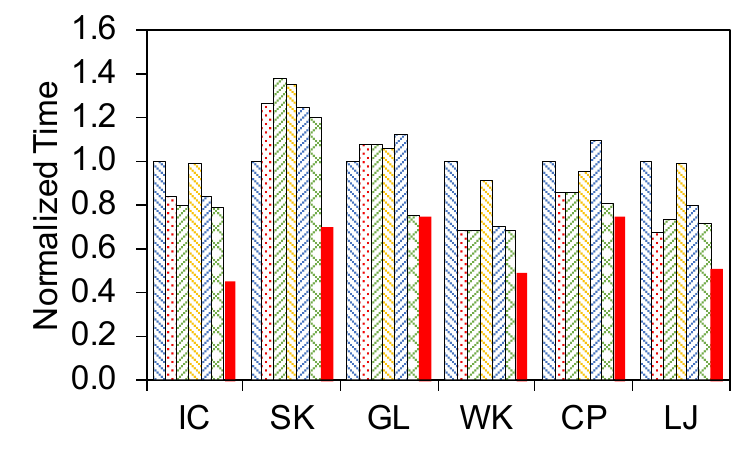}
    \caption{PHP}
  \end{subfigure}
  \vspace{-0.1in}
  \caption{The comparison of runtime}
  \label{fig:overall_runtime}
  \vspace{-0.1in}
\end{figure*}

\begin{figure*}[t]
\vspace{-0.05in}
  \centering
  \includegraphics[width=0.5\textwidth]{Expr/runtime2/tuli.pdf}
  \\
  \begin{subfigure}{0.24\textwidth}
    \includegraphics[width=\linewidth]{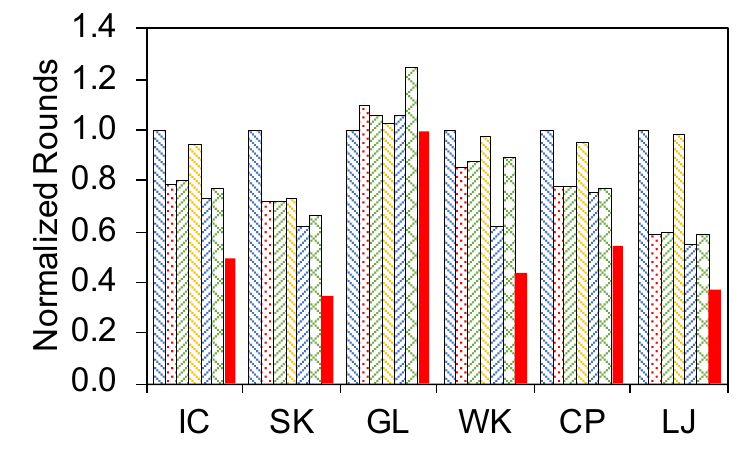}
    \caption{PageRank}
  \end{subfigure}
  \hspace{\fghs}
  \begin{subfigure}{0.24\textwidth}
    \includegraphics[width=\linewidth]{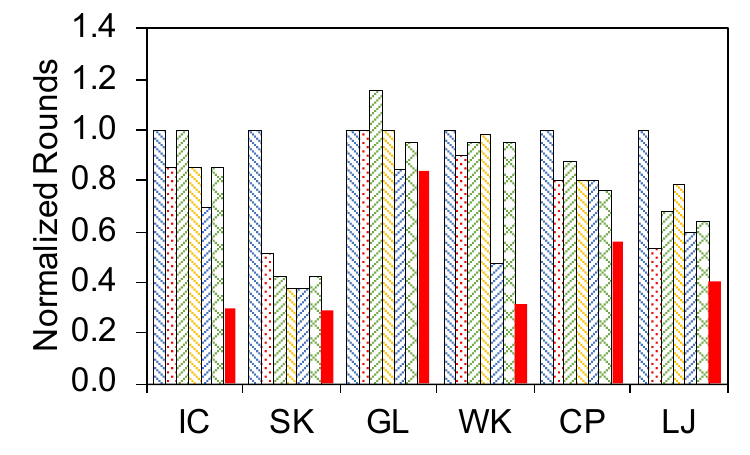}
    \caption{SSSP}
  \end{subfigure}
  \hspace{\fghs}
  \begin{subfigure}{0.24\textwidth}
    \includegraphics[width=\linewidth]{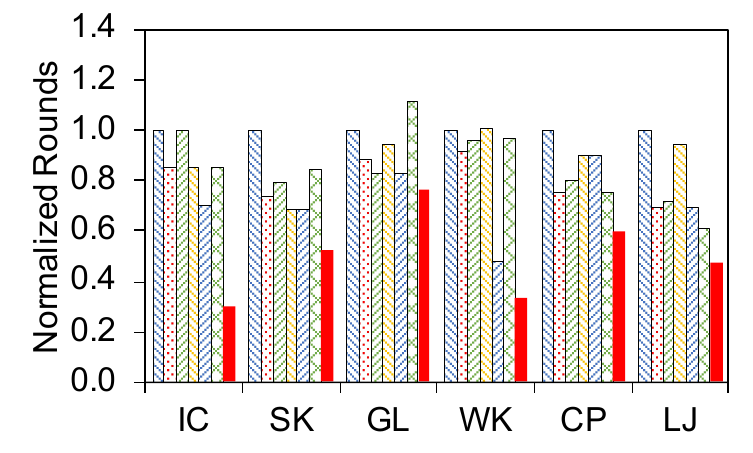}
    \caption{BFS}
  \end{subfigure}
  \hspace{\fghs}
  \begin{subfigure}{0.24\textwidth}
    \includegraphics[width=\linewidth]{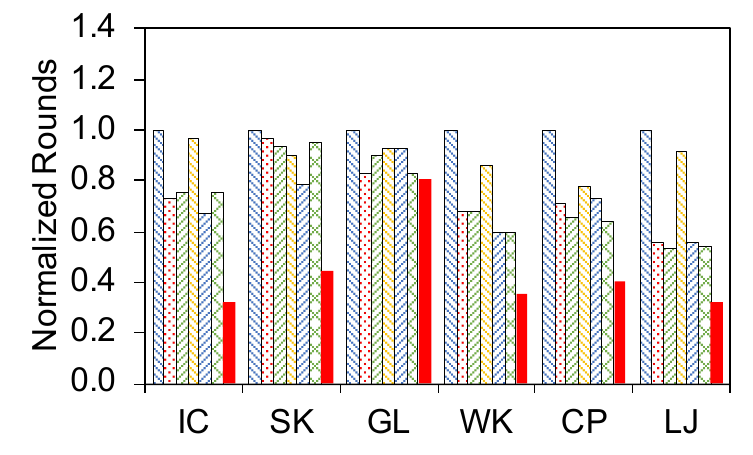}
    \caption{PHP}
  \end{subfigure}
  \vspace{-0.1in}
  \caption{The comparison of iteration rounds}
  \label{fig:overall_iter_round}
  \vspace{-0.15in}
\end{figure*}

\etitle{Competitors}.
We compare \go with the 6 graph ordering methods listed, Default, Degree Sorting, Hub Sorting \cite{zhang2016optimizing}, Hub Clustering \cite{balaji2018graph}, Rabbit \cite{arai2016rabbit}, and Gorder \cite{wei2016gorder}. The default order employs the original IDs as the processing order. 

\subsection{Overall Performance}

\eat{We first assessed \go against competitors by comparing runtime and iteration counts for graph algorithms on Table \ref{tab:data}'s graphs, with results {\em Normalized} in Fig. \ref{fig:overall_runtime} and Fig. \ref{fig:overall_iter_round}, respectively, using Default order as the baseline (unit time 1).}
We first compare \go with the competitors in terms of the runtime and number of iteration rounds for each graph algorithm executed on %all real-world datasets 
graphs in Table \ref{tab:data}. The {\em Normalized} runtime and iteration rounds results are reported in Fig. \ref{fig:overall_runtime} and Fig. \ref{fig:overall_iter_round}, respectively,
where the result of Default order is treated as the baseline, i.e., the Default finishes in unit time 1. 
% 从图\ref{fig:compare_runtime}中可以看出，所有的重新排序算法几乎在所有的数据集上都取得了性能的提升。从整体上来看，竞争者中Rabbit和Gorder对图分析算法的性能提升更加明显，这是由于他们相对于其它简单的重排序算法考虑了更加复杂的局部性优化。然而，他们都忽视了消息传播对图分析算法的影响，正如前文分析的，我们在GoGraph重排序过程中同时考虑了局部性和消息传播的影响，因此, 
It can be seen from Fig. \ref{fig:overall_runtime}, 
\go outperforms others in all cases.
Specifically, \go achieves 
2.10× speedup on average (up to 3.33× speedup) over Default, 
1.66× speedup on average (up to 2.75× speedup) over Degree Sorting, 
1.85× speedup on average (up to 3.24× speedup) over Hub Sorting, 
1.93× speedup on average (up to 3.34× speedup) over Hub Clustering, 
1.80× speedup on average (up to 2.42× speedup) over Rabbit, and 
1.62× speedup on average (up to 2.68× speedup) over Gorder.
On the other hand, in the measurement of the number of iterative rounds as shown in Fig. \ref{fig:overall_iter_round}, \go incurs the least number of iteration rounds on most tested conditions.  Specifically, \go reduces the number of iteration rounds on average by 52\% (up to 71\%) compared with Default,  
39\% (up to 65\%) compared with Degree Sorting,  
40\% (up to 70\%) compared with Hub Sorting, 
45\% (up to 68\%) compared with Hub Clustering, 
32\% (up to 57\%) compared with Rabbit, 
and 39\% (up to 67\%) compared with Gorder.
It is worth noting that \go does not appreciably reduce the number of iteration rounds in the PageRank result for the GL graph, demonstrating that the default order of GL graph is naturally a well-defined processing order. The gain of the corresponding \go %on runtime 
on the GL graph comes from the reduction of cache miss. %consideration of localization.
Another observation is that Orignal, Degree Sorting, Hub Sorting, and Hub Clustering exhibit similar trends in runtime and number of iteration rounds, whereas Rabbit and Gorder do not. This distinction is because the former focuses on the effect of vertex processing order, while the latter concentrates on CPU cache optimization. \go takes both into account in a comprehensive way. %\red{check}

\subsection{Convergence comparison}\label{sec:expr:convergenc}
% 随着迭代的进行，数据集中的每个顶点状态值逐渐趋向收敛。我们通过计算不同运行时刻的顶点状态值与收敛状态的距离来验证GoGraph在加速迭代方面的效果。如图7所示，我们在不同reorder算法得到的数据集上运行PageRank和SSSP进行测试，每个子图包含的折线代表了不同的reorder算法。每条折线由许多个散点组成，其中横坐标表示迭代一轮的运行时间，纵坐标表示当前顶点状态的向量与收敛状态的向量之间的L1范数距离。
To evaluate the effect of \go's reordering, we compared the convergence rates of \go and its competitors.
Our evaluation consists of running the PageRank and SSSP algorithms on CP and LJ graphs, both of which use different reordering algorithms.
We use the absolute difference between 
the sum of vertex state values at convergence and
the sum of all vertex state values at time $t$ to represent the distance to convergence, mathematically expressed as:
% the difference between the converged state and the current state of the vertices
% the L1 norm distance between the vector of current vertex states $\mathbf{V_{t}}$ at instance $t$ and the vector of convergence states $\mathbf{V_{conv}}$
% \begin{equation*}
% ||\mathbf{V_{conv}} - \mathbf{V_{t}}||_1 = \left|\sum_{v\in V} x_{v_{conv}} - \sum_{v\in V} x_{v_{t}}\right|,
% \end{equation*}
%\begin{equation*}
$dist_t = \left|\sum_{v\in V} x^*- \sum_{v\in V} x_t\right|$.
%\end{equation*}
Its trend over time is shown in Fig. \ref{fig:convergence}.
We can see that \go achieves convergence faster in all cases. 
In achieving the same converged state, the \go algorithm consumes only 59\% of the average time used by other algorithms (with a minimum requirement of 37\%).
This efficiency is attributed to the advantages gained in each iteration of the vertex processing order constructed by \go.

\begin{figure*}[!ht]
% \vspace{-0.1in}
   \centering
    \includegraphics[width=0.6\textwidth]{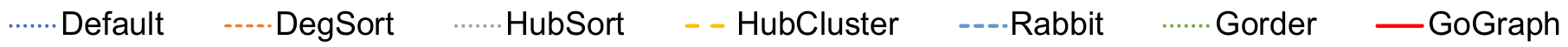}
\\
   \newcommand{\imgwpr}{0.2}
   \subfloat[PageRank on CP]{\includegraphics[width = 0.22\textwidth]{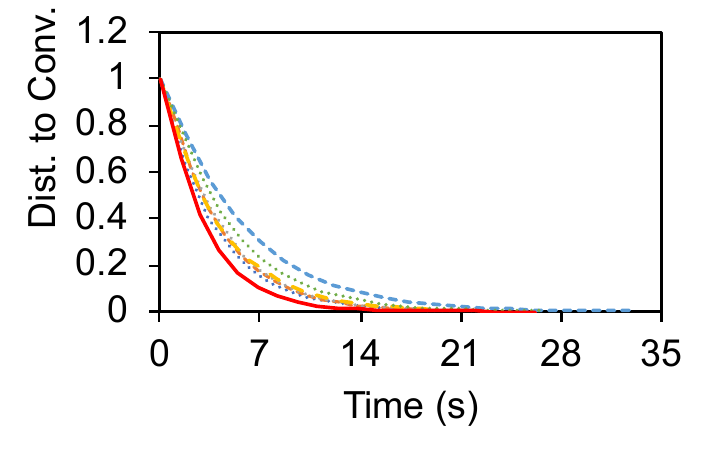}\label{fig:con_pr_2}}\hspace{0.08in}
   % \hfill
   \subfloat[PageRank on LJ]{\includegraphics[width = 0.22\textwidth]{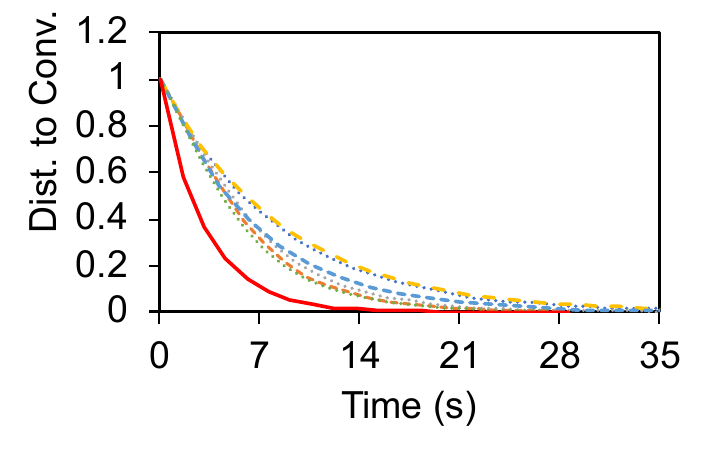}}\hspace{0.08in}
   % \hfill
   \subfloat[SSSP on CP]{\includegraphics[width = 0.234\textwidth]{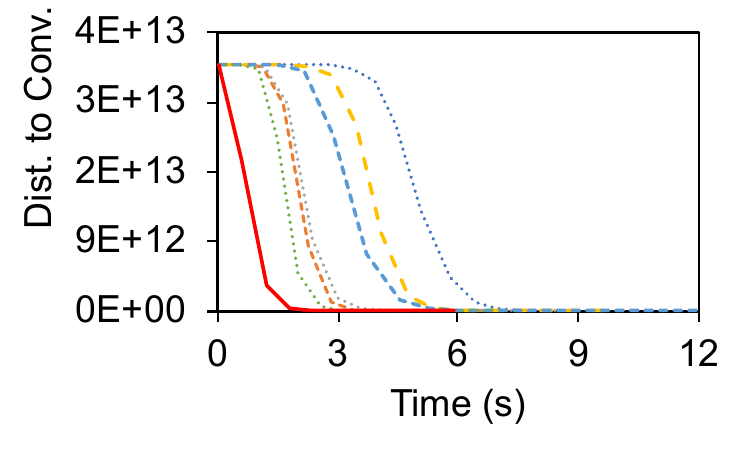}}\hspace{0.1in}
   \subfloat[SSSP on LJ]{\includegraphics[width = 0.247\textwidth]{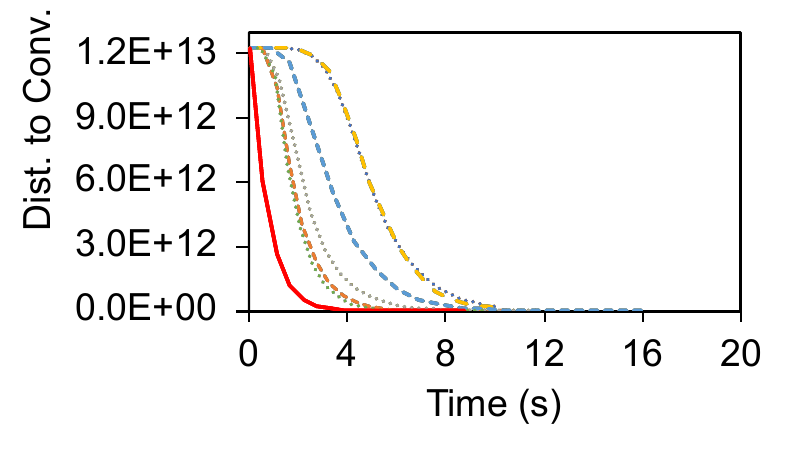}}%\hspace{0.18in}
   % \\
   % \subfloat[PHP on indochina]{\includegraphics[width = \imgw\textwidth]{Expr/convergence/web-indochina-2004_php.pdf}\label{fig:con_php_1}}\hspace{0.18in}
   % \subfloat[PHP on sk-2005]{\includegraphics[width = \imgw\textwidth]{Expr/convergence/sk-2005_php.pdf}\label{fig:con_php_2}}\hspace{0.18in}
   % \subfloat[PHP cit-Patents]{\includegraphics[width = \imgw\textwidth]{Expr/convergence/cit-Patents_php.pdf}\label{fig:con_php_3}}
   \vspace{-0.05in}
    \caption{The comparison of convergence speed}\label{fig:convergence}
    \vspace{-0.15in}
\end{figure*}

\subsection{Impact of processing order in improving Async mode}
% Impact of Async. \& Processing order}%对异步计算的加速
%为了证明processing order对加速异步迭代计算的影响，我们在不同数据集上分别测试了pageran和sssp的执行时间，在同步，异步，和异步采用reorder的情况下。图8展示了同步和异步在默认的处理顺序下，既按数据集顶点ID排序后的processing order，的pagerank和sssp在不同数据集上的执行时间。和reorder后，采用异步计算的执行时间。从图中可以看出异步迭代由于采用顶点的最新状态，其收收敛速度更快，从而减少了迭代计算时间，在pagerank和sssp上分别大约减少了X%-X%和X%-X%。但是大部分不明显只有X-X倍和X-X。除了sssp上的GL，但是而对图进行reorder后，我们可以加速XX倍。

To verify the impact of the asynchronous update mode and the processing orders on accelerating iterative computation, we compare the runtime of PageRank and SSSP on different graphs using synchronous update mode and default processing order (Sync + Def.), asynchronous update mode and default processing order (Async + Def.), asynchronous update mode and the processing order generated by \go (Async + \go). Fig. \ref{fig:accasyc} shows the normalized results. %runtime of PageRank and SSSP on different graphs using synchronous update mode and default processing order (Sync + Def.), asynchronous update mode and default processing order (Async + Def.), asynchronous update mode and the processing order generated by \go (Async + \go).
%从图中可以看出异步可以加速，但是在大多数情况下加速不明显，在pr上有大约5/6的数据集只减少了不到X%运行时间，甚至在WiKi和LJ数据集上值减少了X%运行时间。同样，在sssp上有5/6的数据集只减少了X%运行时间，在GP和LJ上只减少了X%。 而经过reorder后，我们发现平均减少一半以上，即使在LJ上也减少了50%以上，而在IC上减少了80%多的运行时间。
It is shown that asynchronous updating mode can accelerate iterative computation compared with synchronous updating mode. 
%Nevertheless, the impact is constrained and contingent on the dataset and specific algorithm employed.
% but in most cases, the speedup is not obvious. 
%Obviously, under the PageRank algorithm, only the GL graph works well. On the contrary, in the SSSP algorithm, only the IC graph produces a significant gain.
% It is evident that the WK and LJ graphs exhibit lower effectiveness under the SSSP algorithm. Conversely, in the PageRank algorithm, only the IC and GL graphs yield  good gains. 
%Compared to them, 
And \go achieves significant improvements, obtaining 1.56×-6.30× (3.04× on average) speedups. %, while being less sensitive to the datasets and algorithms.
% \red{Only less than X\% runtime is reduced on about 5/6 of the graphs, even worse, on the WiKi and LJ only X\% runtime is reduced. Similarly, on 5/6 graphs, asynchronous mode only reduces X\% SSSP runtime, and on GP and LJ only X\% runtime is reduced. After reordering, the runtime of PageRank and SSSP is reduced by more than 50\%. Even in the worst case, the runtime of iterative computation can still be reduced by X\%, and at best it can be reduced by 80\%.}

\begin{figure}
  \begin{minipage}{0.49\textwidth}
    \centering
    \includegraphics[width=0.6\textwidth]{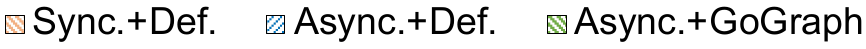}
    \subfloat[PageRank]{\includegraphics[width=0.5\textwidth]{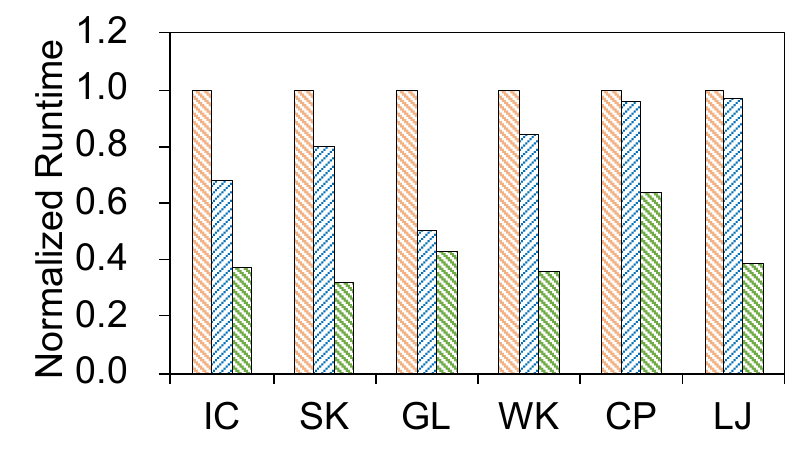}}
    \subfloat[SSSP]{\includegraphics[width=0.5\textwidth]{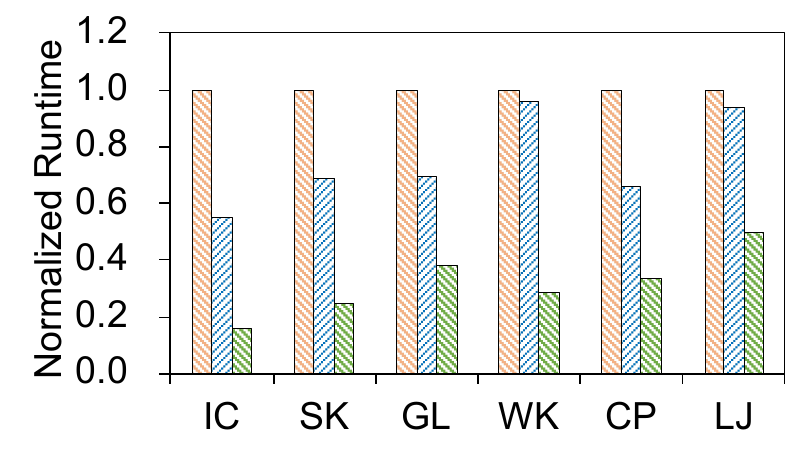}}
    \caption{Impact of processing order in improving Async mode}
    \label{fig:accasyc}
  \end{minipage}
  %\vspace{-0.05in}
\end{figure}

\subsection{CPU Cache Miss}

% \begin{figure}
%    \centering
%    {\includegraphics[width = 0.43\textwidth]{Expr/cache/cache2.pdf}\label{fig:distri_pagerank-1}}
%     \caption{Cache miss comparison.}\label{fig:cachemiss}
%     \vspace{-0.15in}
% \end{figure}

% \begin{figure}
%    \centering
%    \includegraphics[width = 0.43\textwidth]{Expr/cache-partition/cache_par.pdf}
%     \caption{The impact of graph partition on \go.}\label{fig:cachemisspar}
%     \vspace{-0.15in}
% \end{figure}

\begin{figure}[htbp]
    \begin{minipage}[t]{0.24\textwidth}
        \centering
        \includegraphics[width=\textwidth]{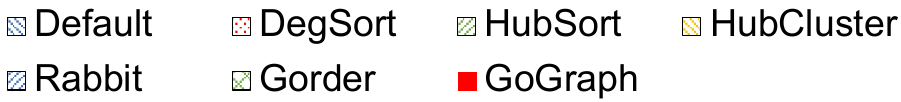}
        \includegraphics[width=\textwidth]{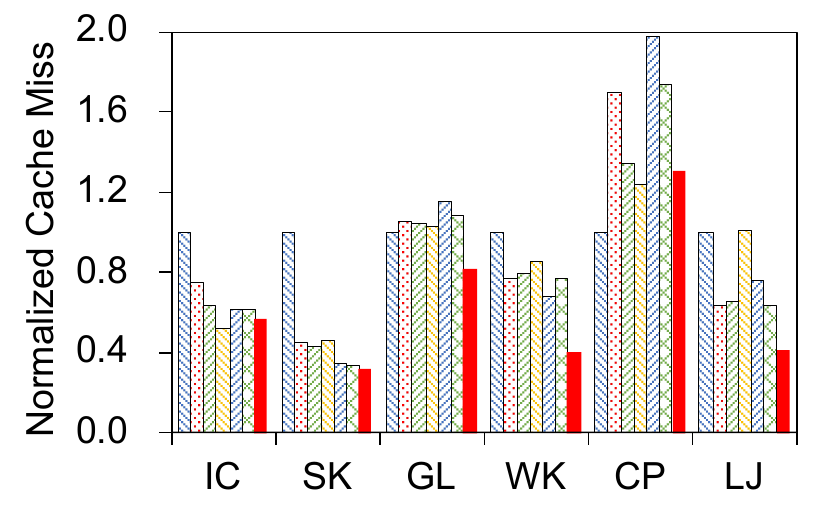}
        \caption{Cache miss comparison}\label{fig:cachemiss}
    \end{minipage}
    \vspace{-0.08in}
    \begin{minipage}[t]{0.24\textwidth}
        \centering
        \includegraphics[width=0.8\textwidth]{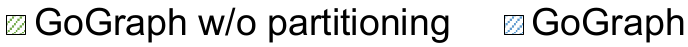}\\
        \includegraphics[width=0.95\textwidth]{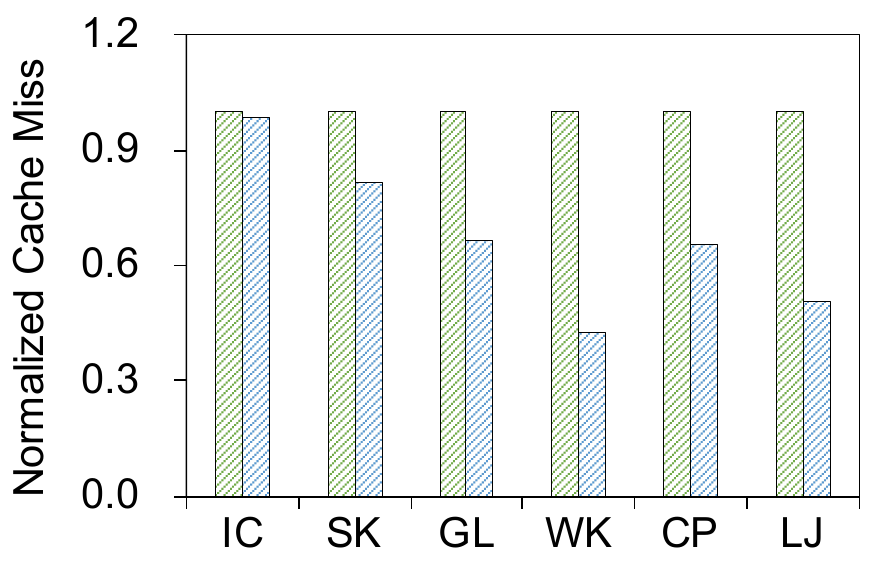}
        \caption{\gr{The impact of partition on cache miss}}\label{fig:cachemisspar}
    \end{minipage}
    \vspace{-0.1in}
\end{figure}

To figure out the effect of \go on reducing cache miss, %we evaluate \go and its competitors on all graphs. %tests datasets. 
we ran PageRank algorithm on all graphs and recorded the cache misses. %for graph analysis five times and recorded the average number of cache misses, 
The results are shown in Fig. \ref{fig:cachemiss}. 
% 从表中的数据可以看出，几乎在所有的测试图分析算法上，我们方法总会导致最少的cache miss数量，这意味着我们将更少更少的内存读取的IO开销。因此，这意味这图\ref{runtime}中GoGraph对应的图分析算法的性能改进来源于重新排序后局部性改进和随之而来的缓存未命中减少。
% As we can see, \go consistently exhibits a lower number of cache misses compared to its competitors in most cases
\gr{Compared to its competitors, \go can reduce the cache miss by 30\% on average}, which means that we incur less I/O overhead when reading graphs from memory.
%thus speeding up the iterative computation. 
This is attributed to our consideration of localization when dividing the graphs.
%in the division of subgraphs phase, the result of improved localization after reordering, and the resulting reduction in cache misses.

% 为了证明图划分对于减少cache miss的作用
\gr{To evaluate the impact of graph partitioning on reducing cache misses, we recorded the cache misses of PageRank on different processing orders obtained by \go with and without partitioning. 
%To analyze whether the graph partitioning operation effectively contributes to \go's reduction of cache misses, we conducted a set of comparative experiments on datasets. The first group excluded the graph partitioning operation but retained other vertex reordering perations in \go, and the other group utilized \go for vertex reordering. For a clearer comparison, we have normalized the results of GoGraph without graph partitioning to 1. 
The results are depicted in Fig. \ref{fig:cachemisspar}. It illustrates that the graph partition contributes to reducing %1.6× 
33\% (up to %2.35× 
58\%) cache misses.} %reduction in cache miss for the graph dataset after reordering.

\subsection{Efficiency of Metric Function}
%gograph的设计基于最大化M的值。为了verify我们的claim that a higher \M value will result in
The design of \go revolves around maximizing the value of $\M(\cdot)$, in that the higher the value of $\M(\cdot)$, the fewer the number of iteration rounds.
% \go is designed for maximizing the $\M(\cdot)$ value since a higher $\M(\cdot)$ value will result in fewer iteration rounds. 
To verify this claim, we record $\M(\cdot)$ values of processing orders of CP graph after applying different reorder methods and the number of iteration rounds of various algorithms on these orders. The results are shown in Table \ref{tab:m-time}. %We find the following.
%$\M$ can efficiently measure the efficiency of processing order in accelerating iterative computation.

\begin{table}[h]
  \centering
  \caption{Metrics and Iteration Rounds of Various Algorithms After Applying Different Reorder Methods on CP.}
  \label{tab:m-time}
  \begin{tabular}{ccccccc}
    \toprule
    \multirowcell{2}[-3pt]{Reorder \\ method} & \multirow{2}{*}{$\M(\cdot)$} & \multirow{2}{*}{$\frac{\M(\cdot)}{|E|}$} & \multicolumn{4}{c}{Number of iteration rounds} \\
     \cmidrule(lr){4-7}
     & & & PageRank & SSSP & BFS & PHP \\
    \midrule
    Default & 1,302,313 & 0.07 & 99 & 25 & 36 & 67 \\
    HubCluster & 2,303,977 & 0.13 & 94 & 20 & 34 & 52 \\
    DegSort & 3,623,082 & 0.20 & 77 & 20 & 25 & 48 \\
    HubSort & 3,691,804 & 0.20 & 77 & 22 & 26 & 44 \\
    Gorder & 5,875,924 & 0.32 & 76 & 19 & 22 & 43 \\
    Rabbit & 8,883,616 & 0.49 & 75 & 20 & 25 & 49 \\
    \go & 13,871,315 & 0.76 & 54 & 14 & 17 & 27 \\
    \bottomrule
    \vspace{-0.2in}
  \end{tabular}
\end{table}

It can be seen that \romannumeral1) the larger the $\M(\cdot)$ value of the processing order is, the fewer iteration rounds the algorithm requires, which is in line with our claims and expectations. %Although the $\M(\cdot)$ values of the processing orders generated by Degree Sorting, Hub Sorting, and Gorder are increasing one by one, the number of iteration rounds on PageRank, SSSP, BFS, and PHP fluctuates slightly. 
%This is because the gap of $\M(\cdot)$ between them is tiny, with only 12\% (i.e., $32\%-20\%$) of edges. Although Rabbit can obtain a larger $\M(\cdot)$ value than other competitors, Rabbit makes some edges negative, which play an important role in SSSP, BFS, and PHP, e.g., the edges connected with high-degree vertices. 
%Conversely, \go extracts high-degree vertices and isolated vertices and considers them individually.
%虽然degsort，hubsort，goroder生成的order的M值越来越大，但是他们的迭代轮数在pr，sssp，bfs和php上却表现的相差不是很大，这是因为虽然越来越大，但是他们的差距非常小，且只占总边数比例的X%-X%。而rabbit虽然能够获得一个较大的m值相比于其他算法，但是rabbit将一些重要的边弄成了negative，这些边在sssp，BFS和PHP中起了比较重要的作用。因此为了我们将探索一个更好的metric function来更精确衡量迭代计算的好坏。，并没有考虑大度点，导致一些与大度点相连的边变成negative的。而这些边往往更有利于顶顶点状态的传播。而在gograph中，我们将大度点最后插入，让更多的与大度点相连的边为positive
%\sstab
\romannumeral2) The $\M(\cdot)$ value of processing order produced by \go is the largest and the number of iteration rounds always are smallest. It means that the metric function $\M(\cdot)$ is effective for measuring the efficiency of processing order in accelerating iteration computations.

\subsection{Memory Usage}

\eat{Our approach enhances the performance of iterative computations by reducing the number of iterations. Compared with the baseline (synchronous update on default graph processing order), the asynchronous iterative computation does not introduce additional memory overhead in theory. %our method should not lead to an increase in memory usage. 
To assess the memory overhead of our method, we evaluated the memory usage for three scenarios, synchronous update on default processing order (Sync. + Def.), asynchronous update on default processing order (Async. + Def), and asynchronous updates on a order reordered by \go  (Async. + \go). The results are shown in Fig. \ref{fig:mem-usage}. %, with the results normalized against the baseline for comparison. 
It can be seen that the memory usage of all three methods is comparable. 
%Our approach enhances the performance of iterative computations by reducing the number of iterations. 
%Compared with the baseline (Sync. + Def.), the asynchronous iterative computation does not introduce additional memory overhead. Because 
\go enhances the performance of iterative computation only by rearranging the processing order without additional data structures. However, the synchronous version incurs slightly higher memory overhead due to the need to record the current state and the previous state of the vertices.
}

Our method boosts iterative computation performance by reducing iteration counts without theoretically increasing memory overhead, compared to the baseline (synchronous update with default graph order). To evaluate memory overhead, we examined memory usage in three scenarios: synchronous update with default order (Sync. + Def.), asynchronous update with default order (Async. + Def), and asynchronous update reordered by \go (Async. + \go), shown in Fig.  \ref{fig:mem-usage}. Memory usage across these methods is similar. \go improves iterative computation efficiency by reordering processing without extra data structures. However, the synchronous approach slightly increases memory overhead due to recording vertices' current and previous states.

\begin{figure}[ht]
  \begin{minipage}{0.49\textwidth}
    \centering
    \includegraphics[width=0.6\textwidth]{Expr/lastfigure/tuli.pdf}
    \subfloat[PageRank]{\includegraphics[width=0.5\textwidth]{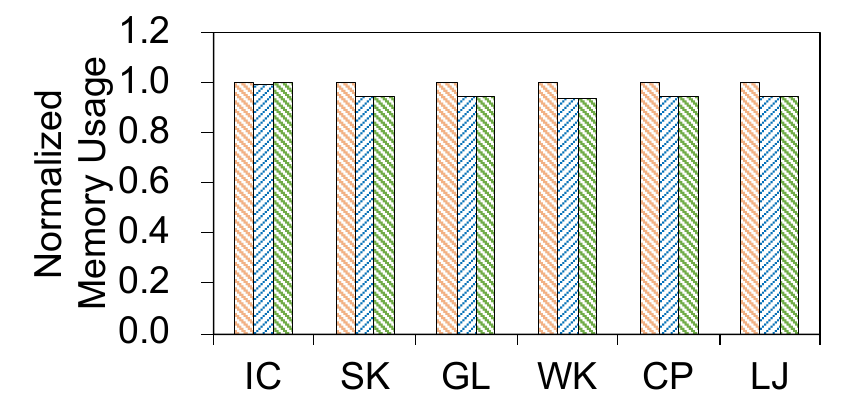}}
    \subfloat[SSSP]{\includegraphics[width=0.5\textwidth]{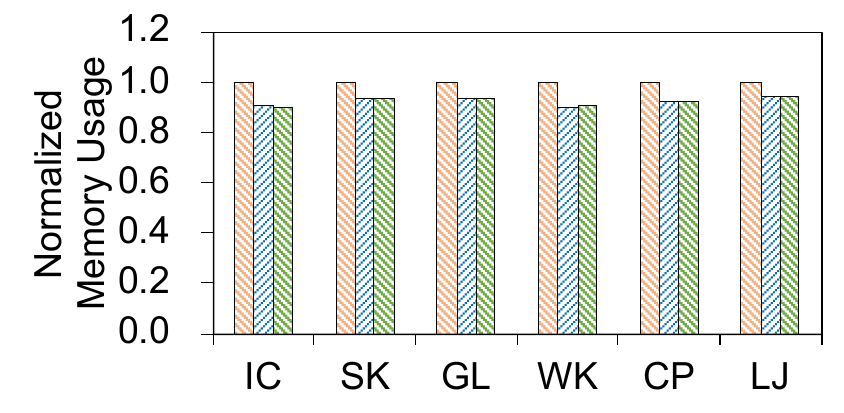}}
    \caption{\gr{Memory usage of different iterative computations}}
    \label{fig:mem-usage}
  \end{minipage}
  \vspace{-0.15in}
\end{figure}

\subsection{Average Degrees}

% \begin{figure}[ht]
%   \begin{minipage}{0.5\textwidth}
%     \centering
%     \includegraphics[width=0.55\textwidth]{Expr/avgdeg/tuli.pdf}
%     \\
%     \subfloat[Runtime]{\includegraphics[width=0.5\textwidth]{Expr/avgdeg/2.pdf}
%     }
%     \subfloat[Number of Iteration Rounds]{\includegraphics[width=0.5\textwidth]{Expr/avgdeg/1.pdf}
%     }
%     \vspace{5pt}
%     \caption{The impact of different average degrees}
%     \label{fig:diff-deg}
%   \end{minipage}
%   % \vspace{-0.15in}
% \end{figure}

% \blue{To test the effects of average degrees of graphs on the performance of \go, we generated different graphs with varying average degrees, 2,  4, 6, 8, using PaRMAT\cite{rmat}, where the number of vertices in each graph is fixed to 1,000,000. We reordered them with different reorder methods and performed the PageRank algorithm on them. %to test the number of iterations and iteration time for each dataset. 
% Fig. \ref{fig:diff-deg} shows the number of iteration rounds and runtime of PageRank on different datasets.}

% \blue{It can be observed that \go outperforms other reordering methods on both the number of iteration rounds and runtime when varying the average degrees of graphs. As the average degree of the graph increases, the runtime of the PageRank also increases. This is because the size of the graph becomes larger and the running time of each iteration becomes longer, while the number of iteration rounds does not change significantly.}

\gr{To evaluate the impact of average graph degrees on \go's performance, we generated a series of graphs with different average degrees (2, 4, 6 and 8) using the Barabasi-Albert model\cite{barabasi1999emergence} in NetworkX, each with 1,000,000 vertices. We applied different reordering methods and ran PageRank on these graphs. Fig. \ref{fig:diff-deg} shows PageRank's runtime and iteration counts across datasets, showing \go's superior performance in both aspects with varying average degrees.}
\begin{figure}[ht]
\vspace{-0.1in}
  \begin{minipage}{0.49\textwidth}
    \centering
    \includegraphics[width=0.88\textwidth]{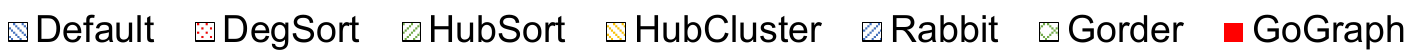}
    \\
    %\vspace{-10pt}
    \subfloat[Runtime]{\includegraphics[width=0.48\textwidth]{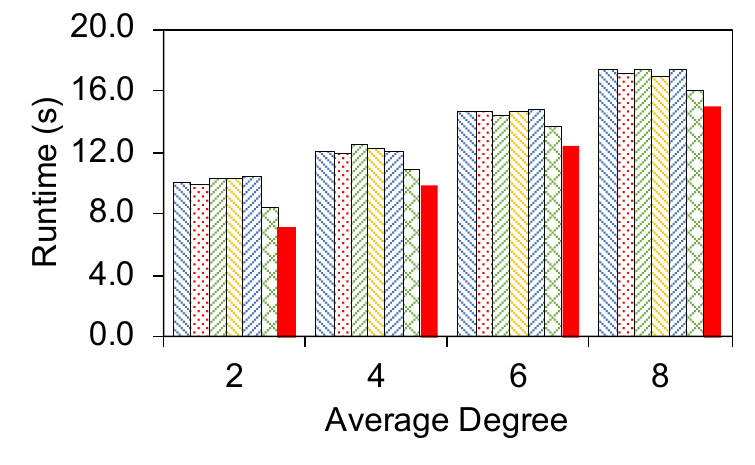}}
    \subfloat[Number of Iteration Rounds]{\includegraphics[width=0.48\textwidth]{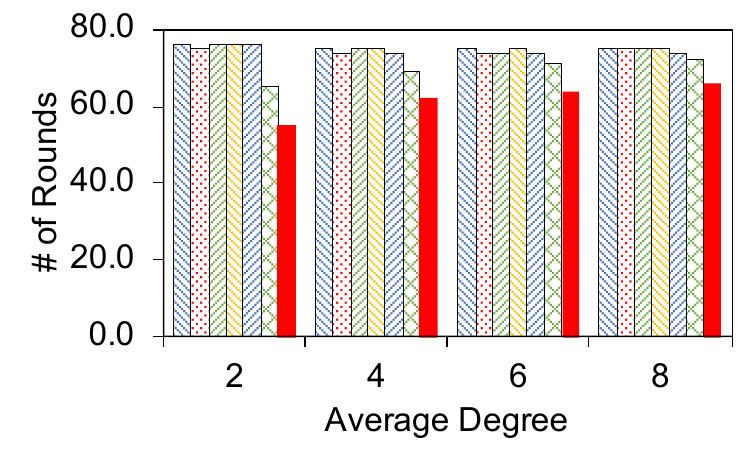}}
    % \vspace{5pt}
    \caption{The impact of different average degrees}
    \label{fig:diff-deg}
  \end{minipage}
   \vspace{-0.15in}
\end{figure}
% 在图中可以观察到，迭代轮数和迭代时间与平均度数的增加呈现出正相关的趋势。随着平均度数的提高，每个顶点的邻居数量也增加，导致迭代计算的复杂度随之增加。\go 算法在所有 reorder 算法中表现出色。这是因为 \go 能够有效地促使重新排序后的图在每一轮迭代中更充分地利用邻居点更新后的状态值，因此 \go 能够稳定地实现更好的加速效果。
%From the figure, 
% It can be observed that \go outperforms other reordering methods in terms of both runtime and the number of iteration rounds when varying the average degrees of graphs. 
\eat{As the average degree of the graph increases, the runtime of PageRank also increases. This is due to the enlargement of the graph's size, resulting in longer running times for each iteration, while the number of iteration rounds remains relatively unchanged. However, on these artificial graphs, all the graph reordering methods do not perform as well as on real graphs.
% This is because the default processing order of these generated graphs is better than that of real graphs.
This is because the default processing order of these generated graphs is better than that of real graphs.
%We attribute this to the better default processing order of these generated graphs compared to real graphs, 
Hence, the improvement yielded by \go and other reordering methods are not as significant as in real graphs.
%We attribute this to the better default processing order of these generated graphs compared to real graphs, hence the improvement from \go is not as pronounced.
}

Increased average degrees led to longer PageRank runtimes due to larger graph sizes, though iteration counts stayed similar. However, reordering methods were less effective on these synthetic graphs compared to real ones, as the default order of generated graphs is already more optimal than that of real graphs, diminishing the enhancements from \go and others.

\subsection{Partition Methods}

\gr{In \go, we use the graph partition method introduced in Rabbit (Rabbit-Partition) by default. 
To evaluate the effectiveness of various graph partitioning methods in \go, we employed Metis\cite{metis}, Louvain\cite{louvain}, and Fennel\cite{tsourakakis2014fennel} for graph partitioning.
Fig. \ref{fig:diff-part} presents the normalized runtime and iteration counts of PageRank on graphs reordered by \go using these methods, with Rabbit-Partition as the baseline. }

\begin{figure}[ht]
\vspace{-0.15in}
  \begin{minipage}{0.49\textwidth}
    \centering
    \includegraphics[width=0.56\textwidth]{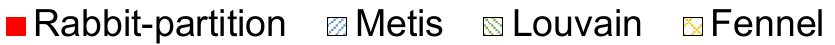}
    \\
    %\vspace{-10pt}
    \subfloat[Normalized Runtime]{\includegraphics[width=0.48\textwidth]{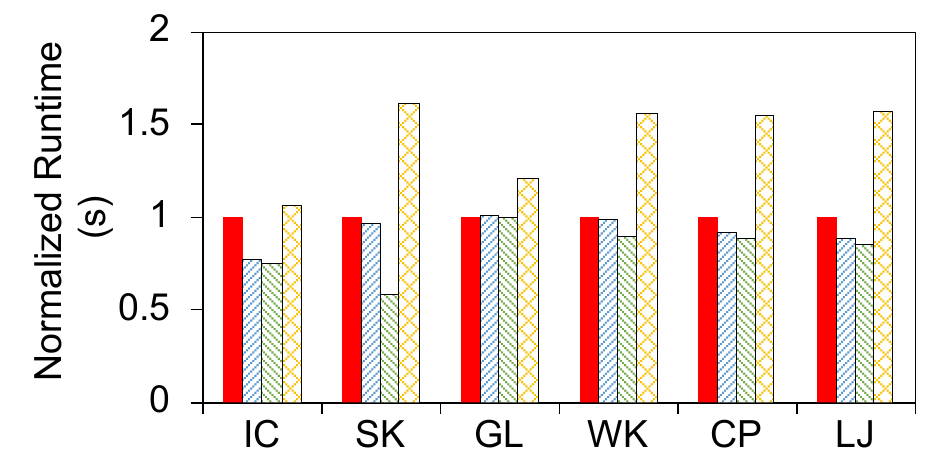}}\ \
    \subfloat[Normalized Rounds]
    {\includegraphics[width=0.48\textwidth]{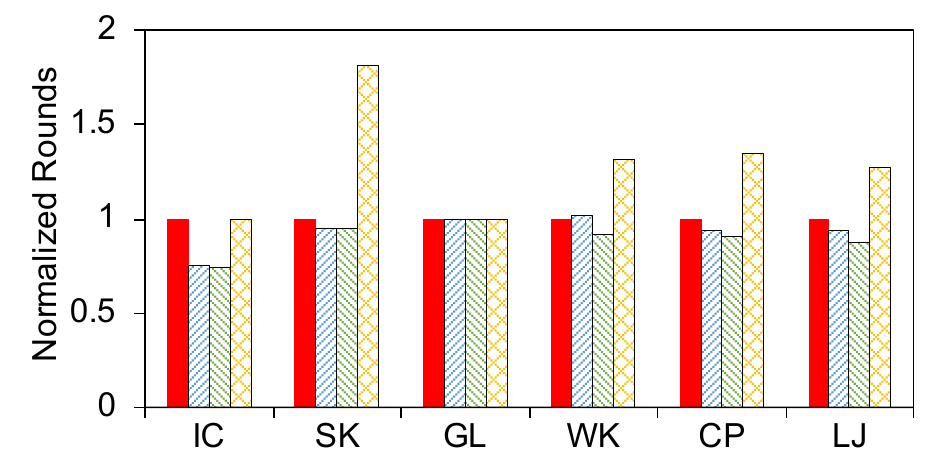}}
    % \vspace{5pt}
    \caption{The impact of different partition methods}
    \label{fig:diff-part}
  \end{minipage}
   \vspace{-0.1in}
\end{figure}

\eat{The experimental results reveal that employing different graph partitioning methods results in different runtime and the number of iteration rounds when executing graph algorithms on reordered graphs. Generally, Rabbit-partition, Metis, and Louvain exhibit comparable performance, whereas Fennel performed worse. %underperforms relative to the others. 
This is because Fennel is a stream-based partitioning method, which relies on partial knowledge of the graph for making partitioning decisions. It means that a better graph partitioning method helps \go improve performance. %This finding points out the importance of selecting an effective graph partitioning method to enhance \go's performance. 
}
The results indicate that different partitioning methods impact the runtime and iteration counts for graph algorithms on reordered graphs. Rabbit-Partition, Metis, and Louvain showed similar performance, while Fennel underperformed due to its stream-based approach, which makes decisions with partial graph knowledge. Thus, a superior partitioning method enhances \go's performance.

\vspace{-0.05in}
\section{Related Work}\label{sec:relate}
\vspace{-0.05in}
% Our work is related to reordering methods and CPU cache efficiency. This section provides a summary of the related work.

\stitle{Optimizations on synchronous iteration.}
%异步计算在加速迭代计算上具有一定的优势，已经有一些工作尝试在异步计算模式下加速迭代计算。比如Priter和Maiter在每轮迭代优先选择可能会促进算法收敛的点执行迭代计算，而不是所有的点，从而尝试避免一些低效计算。然而这种方法需要用户去调试选择顶点的方式，因为不同的迭代算法判断顶点对算法收敛贡献的方式不同。FBSGraph则采用了一种forward and backward sweeping execution framework to process vertices, 这避免了顶点状态在传给其在processing order前方出邻居较慢的问题。PathGraph在FBSGraph的基础上进一步对顶点状态的路径进行了优化，使顶点的状态能够沿着顶点的path传播。但是这些方法并没有从本质上揭示异步迭代计算加速的根本原因，GoGraph在进一步揭示异步加速迭代计算原因的同时，通过重排序，努力最大化异步计算的优势。
There have been efforts to accelerate iterative computing by exploring asynchronous computing modes. Priter \cite{zhang2011priter} and Maiter \cite{zhang2013maiter} preferentially select vertices in each iteration\eat{that may speedup algorithm convergence to perform iterative computations instead of all vertices}{, aiming to expedite algorithm convergence by performing iterative computations on specific vertices instead of all vertices.}, thereby trying to avoid some inefficient computation. However, this method requires user-defined vertex selection strategies, which is not a trivial work. FBSGraph \cite{fbsgraph} employs a forward and backward sweeping execution framework for vertex processing, addressing the issue of slow propagation of vertex states when their outgoing neighbors precede them in the processing order. PathGraph \cite{pathgraph} further optimizes the path of vertex state propagation, enhancing the efficiency of state propagation. %However, these methods do not essentially reveal the fundamental reasons for the acceleration of asynchronous iterative computations. 
\go tries to maximize the advantages of asynchronous computations by rearranging graph processing order based on revealing the reason why asynchronous updating mode accelerates iterative computations.
%Many iterative graph computing systems aim at prioritizing cache efficiency. PathGraph\cite{pathgraph} utilizes a tree-based partitioning method and reorders the vertices in each partition using DFS to maximize sequential access and minimize random access, thereby enhancing access locality and cache hit rate. FBSGraph\cite{fbsgraph} proposes a forward and backward sweeping execution framework to process vertices sequentially along the graph path, such that the vertices of a path that were swapped into the cache in the current round have a high probability of being reaccessed in the subsequent round. It increases cache effectiveness by decreasing cache replacement. \go utilizes BFS method to select neighbor vertices to insert into the processing order, in order that assigning as many consecutive storage locations as possible to a vertex's adjacent vertices, thereby enhancing locality.
% , decreasing cache substitution, and improving cache efficiency.

\stitle{Graph reordering.}
Vertex reordering
has been a focal point in graph data preprocessing to enhance memory access efficiency through increased vertex locality.
% has received a lot of attention for pre-processing graph data. There has been much research into reordering algorithms in order to improve memory access efficiency by increasing the locality of vertices. 
To improve the temporal and spatial locality, Gorder~\cite{wei2016gorder} uses a slide window to compute the score between the ordered vertices and the unordered vertices. The larger the score, the more frequently the unordered vertices will be accessed after the ordered vertices, from which an ordering algorithm can be deduced. Rabbit~\cite{arai2016rabbit} maps the more frequently accessed vertices close to the L1 cache, thus reducing the overhead of swapping cache lines. Hub Clustering~\cite{balaji2018graph} assigns a contiguous range of subscripts to hub vertices whose degrees are larger than the average degree at the front of the graph data array. Since the neighbors of the high-degree vertices are likely to overlap, storing them together in memory decreases the frequency of cache line swapping. Hub Sorting (also known as frequency based clustering)~\cite{zhang2016optimizing} is a lightweight reordering method that extracts hub vertices, arranges them in descending order, and then swaps them with the vertices with continuous subscripts at the front of the data array. Thus, the subscripts of non-hub vertices can be preserved as much as possible in order to reduce the cost of reordering operations and improve the locality of the power law graph. 
% They reorder the vertices to improve graph locality, thereby increasing cache hit rate and accelerating iteration speed. \go focuses on determining the processing order of vertices. This enables a larger number of vertices to update their states based on the updated state of neighboring vertices within the current iteration round.
They reorder vertices to enhance graph locality, boost cache hit rates, and speed up iterations. GoGraph optimizes vertex processing order, allowing more vertices to refresh their states using neighbors' latest states within the same iteration.

\section{Conclusion}\label{sec:conclu}
In this paper, we propose \go, a graph reordering algorithm that establishes a well-formed graph processing order, 
resulting in a reduction of the number of iteration rounds and acceleration of iterative computation. 
% Initially, we propose a metric to measure the efficiency of the vertex processing order in accelerated iterative computations. The metric characterizes the quality of the processing order by counting the number of positive edges.
% Subsequently, \go employs a divide-and-conquer idea to construct the vertex processing order that maximizes metric function value.
% Finally, experiments verify the validity of the vertex processing order organized by \go.
Specifically, we introduce a metric to evaluate the quality of the processing order, based on the count of positive edges. \go employs a divide-and-conquer strategy to optimize this metric, with experimental results validating its effectiveness in reorganizing the vertex processing order.

\vspace{-0.05in}
\section*{Acknowledgment}
\vspace{-0.05in}
This work is supported by The National Key R\&D Program of China (2018YFB1003400), 
The National Natural Science Foundation of China (U2241212, 62072082, 62202088, 62072083, and 62372097), Joint Funds of Natural Science Foundation of Liaoning Province (2023-MSBA-078), Research Grants Council of Hong Kong, China, No.14205520, and Fundamental Research Funds for the Central Universities (N2216012 and N232405-17).

% reducing cache misses and increasing message passing efficiency. \go takes into account both the temporal and spatial locality of the graph data structure and the direction of message propagation, it effectively reduces the graph iteration time, improves message utilization, and decreases the number of CPU cache misses. 
% Our experiments show that in the aforementioned respects, \go is superior to the existing reordering algorithm. Future work will include implementing \go on distributed systems so that they can process large graphs more efficiently.

\bibliographystyle{IEEEtran}  %声明选择的格式: IEEEtran, acm, unsrt
\bibliography{ref} %bib文件名，需要放在同一个文件夹下，否则要在filename前说明路径

\end{document}